\documentclass[letterpaper, 10 pt, conference]{article}  

\title{\LARGE \bf
Data Driven Estimation of Stochastic Switched Linear Systems of Unknown Order
}

\author{Tuhin Sarkar \hspace{5mm} Alexander Rakhlin \hspace{5mm} Munther Dahleh
\thanks{TS,AR,MD are with Massachusetts Institute of Technology, Cambridge, MA 02139
        (email: {\tt\small tsarkar,rakhlin, dahleh@mit.edu})}%
}

	\usepackage[a4paper, margin=1in]{geometry}
	\usepackage{epigraph}
	
	\usepackage{csquotes}
	\usepackage{url}       
	\usepackage{relsize}     
	\usepackage{booktabs}       
	\usepackage{amsfonts}       
	\usepackage{nicefrac}       
	\usepackage{microtype}      
	\usepackage{color}
	\usepackage{scalerel}
	\usepackage[font=small,labelfont=bf]{caption}
	\usepackage{blindtext}
	\usepackage{amssymb}
	\usepackage{mathtools}
	\usepackage{dsfont}
	\usepackage{amsmath}
	\usepackage{amsfonts}
	\usepackage{amsthm}
	\usepackage{float}
	\usepackage{graphicx}
	\usepackage{subcaption}
	\graphicspath{ {images/} }
	\usepackage{letltxmacro}%
	\usepackage{thmtools,thm-restate}
	\usepackage{relsize}
	\usepackage{algorithm}
	\usepackage[noend]{algpseudocode}

	\newcommand{\Hc}{\mathcal{H}}
	
	\newcommand{\Mc}{\mathcal{M}}

\newcommand{\bigO}{\Oc}
\newcommand{\bigOtilde}{\wt{\Oc}}

\usepackage{amsthm}
\newtheoremstyle{mystyle}
  {}
  {}
  {\itshape}
  {}
  {\bfseries}
  {.}
  { }
  {}

\theoremstyle{mystyle}

	\newcommand{\bl}{\Big |}

\DeclarePairedDelimiter{\prn}{(}{)}
\DeclarePairedDelimiter{\nrm}{\lVert}{\rVert}
\DeclarePairedDelimiter{\curly}{\{}{\}}

	\newcommand{\bcF}{{\mathcal{F}}}
	
	\newcommand{\Nc}{{\mathcal{N}}}

	\newcommand{\Ex}{\mathbb{E}}
	\newcommand{\Pb}{\mathbb{P}}
	\newcommand{\Rb}{\mathbb{R}}
	\newcommand{\Oc}{\mathcal{O}}
	
	\newcommand{\Sc}{\mathcal{S}}

	\newcommand{\Ec}{\mathcal{E}}

	\newcommand{\wh}{\widehat}

	\newcommand{\hHc}{\wh{\Hc}}
	
	\newcommand{\hd}{\wh{d}} 
	
	\newcommand{\sys}{(C, \{A_i, p_i\}_{i=1}^s, B)}

    \newcommand{\tA}{\wt{A}}
    \newcommand{\tB}{\wt{B}}
    \newcommand{\tC}{\wt{C}}
    \newcommand{\hU}{\wh{U}}
    \newcommand{\hSigma}{\wh{\Sigma}}
    \newcommand{\hV}{\wh{V}}   
    \newcommand{\hN}{{\wh{N}}}
    
    \newcommand{\tl}{\textbf{l}}
    \newcommand{\HN}{\Hc^{(N)}}

    \newcommand{\Ns}{N^{*}}

    \newcommand{\Hci}{\Hc^{(\infty)}}
    \newcommand{\hHN}{\wh{\Hc}^{(N)}}
    
	\newtheorem{assumption}{Assumption}
	\newtheorem{proposition}{Proposition}
	\newtheorem{example}{Example}
	\newtheorem{definition}{Definition}
	\newtheorem{theorem}{Theorem}
	\newtheorem{remark}{Remark}

	\newcommand{\hp}{\wh{p}}
	\newcommand{\wt}{\widetilde}
	\newcommand{\est}{\wt{E}}

	\newcommand{\trunc}{\wt{T}}
    \newcommand{\ts}{\widetilde{s}_0}
	\newcommand{\subg}{\mathsf{subg}}
	\newcommand{\const}{\alpha}
\makeatletter
\renewcommand*{\eqref}[1]{%
  \hyperref[{#1}]{\textup{\tagform@{\ref*{#1}}}}%
}
\makeatother
\begin{document}

\maketitle
\thispagestyle{empty}
\pagestyle{empty}

\begin{abstract}
We address the problem of learning the parameters of a mean square stable switched linear systems (SLS) with unknown latent space dimension, or \textit{order}, from its noisy input--output data. In particular, we focus on learning a good lower order approximation of the underlying model allowed by finite data. Motivated by subspace-based algorithms in system theory, we construct a Hankel-like matrix from finite noisy data using ordinary least squares. Such a formulation circumvents the non-convexities that arise in system identification, and allows for accurate estimation of the underlying SLS as data size increases. Since the model order is unknown, the key idea of our approach is model order selection based on purely data dependent quantities. We construct Hankel-like matrices from data of dimension obtained from the order selection procedure. By exploiting tools from theory of model reduction for SLS, we obtain suitable approximations via singular value decomposition (SVD) and show that the system parameter estimates are close to a balanced truncated realization of the underlying system with high probability.
\end{abstract}
\section{Introduction}
\label{introduction}
Finite time system identification is an important problem in the context of control theory, times series analysis and robotics among many others. In this work, we focus on parameter estimation and model approximation of switched linear systems (SLS), which are described by
\begin{align}
    \label{switch_lti}
x_{k+1} &= A_{\theta_k}x_k + Bu_k + \eta_{k+1} \\
y_k &= C x_k + w_k \nonumber
\end{align}
Here at time $k$, $x_k \in \Rb^{n}, y_k \in \Rb^{p}, u_k \in \Rb^{m}$ are the latent state, output and input respectively. $\theta_k \in \{1, 2, \ldots, s\}$ is the discrete state, mode or switch with $\eta_{k}, w_k$ being the process and output noise respectively. We assume that $\{\theta_k\}_{k=1}^{\infty}$ is an i.i.d process with $\Pb(\theta_k = i) = p_i$. The goal is to learn $(C, \{p_i, A_i\}_{i=1}^s, B)$ from observed data $\{y_k, u_k, \theta_k\}_{k=1}^N$ when the latent space dimension $n$ is unknown. In many cases $n > p, m$ and it becomes difficult to find suitable parametrizations that allow for provably efficient learning. For the special case of LTI systems, \textit{i.e.}, $s=1$, these issues were discussed in detail in~\cite{sarkar2019finite}. It was suggested there that one can learn lower order approximations of the original system from finite noisy data. To motivate the study of such approximations, consider the following example:
\begin{example}
\label{effective_order}
Let $s=2, p_i = 0.5, |\gamma| < 1$. Consider $M_1 = \{C, A_1 \in \Rb^{n \times n}, B\}, M_2 =\{C, A_2 \in \Rb^{n \times n}, B\}$ given by
\begin{align}
    B &= \begin{bmatrix}
    0 \\
    \vdots \\
    0 \\
    1
    \end{bmatrix}, C = B^{\top},  A_1 = \begin{bmatrix}
    0 & 0 & \hdots & 0\\
    0 & 0 & \hdots & 0 \\
    \vdots & \vdots & \vdots & \vdots \\
    0 & \hdots & 0 & \gamma 
    \end{bmatrix}     \nonumber \\
    A_2 &= \begin{bmatrix}
    0 & 1 & \hdots & 0\\
    0 & 0 & \ddots & 0 \\
    \vdots & \vdots & \vdots & \ddots \\
    a & 0 & \hdots & 0 \\
    \end{bmatrix} 
\end{align}
Assume that $na \ll  1$. This SLS is of order $n$, which may be large. However, it can be suitably modeled by a lower dimensional SLS (``effective'' order is $\leq 2$ and can be checked by a simple computation of $\{CA_i A_jB\}_{i,j=1}^2$). In fact, the SLS with parameters $B = C = 1$ and $A_1 = \gamma, A_2 = 0$ generate the same output. 
\end{example}
The previous example suggests that in many cases the true order is not important; rather a lower order model exists that approximates the true system well. Furthermore, finite noisy data limits the complexity of models that can be effectively learned (See discussion in~\cite{venkatesh2001system}). The existence of an ``effective'' lower order and finite data length motivate the question of finding ``good'' lower dimensional approximations of the underlying model from finite noisy data. 

Classical system identification results provide asymptotic guarantees for learning from data~\cite{ljung1992asymptotic}. Finite sample analysis for linear system identification has also been studied in~\cite{campi2002finite}. However, when considering finite noisy data problems such as those occurring for example in reinforcement learning, these results are either too pessimistic (finite time bounds) due to a worst case analysis or inapplicable due to their asymptotic nature. In this work, we focus on estimation of stochastic switched linear systems when the model order is unknown. We believe such a discussion is important in two aspects: first, typically the hidden state $x_k$ in Eq.~\eqref{switch_lti} is unobserved. Second, as discussed in Example~\ref{sls_example} although the true system may be of a high order, it could also be ``well'' approximated by a much lower order system.

\subsection{Related Work}
\label{rel_work}
The study of switched linear systems (SLS) has attracted a lot of attention~\cite{ezzine1989controllability,branicky2005introduction,sun2006switched} to name a few. These have been used in neuroscience to model neuron firing~\cite{linderman2016recurrent}, modeling the stock index~\cite{fox2011bayesian} and more generally approximate non--linear processes~\cite{westra2011identification} with reasonable accuracy. The theoretical aspects of SLS (deterministic or stochastic) have been studied on several fronts: realizability, model reduction and identification from data. 

The problem of realization, \textit{i.e.}, whether there exists a SLS that satisfies the given data (in the noiseless case), has been studied in~\cite{grossman1995algebraic,petreczky2010realization,petreczky2018} and references therein. Specifically, \cite{grossman1995algebraic} provides a purely algebraic view of realization where the switching is a function of exogenous discrete input symbols. The authors in~\cite{petreczky2010realization} consider the case when discrete events are external inputs and there are linear reset maps that reset the state after switching. They provide necessary and sufficient conditions for existence of a SLS that generates the given output data in response to certain inputs. The question of minimality via the rank of a generalized Hankel matrix is also discussed there. The results for SLS there imply the classical ones for linear time invariant (LTI) systems. These ideas are further extended to generalized bilinear systems in~\cite{petreczky2018}. 

Hybrid systems, where systems have a continuous state and discrete state, are a generalization of SLS and piecewise linear ARX (PWARX) systems. Many approaches have been devised for hybrid system identification: support vector regression(SVR)-based approach~\cite{lauer2019hybrid}, the clustering-based approach~\cite{ferrari2003clustering}, the mixed integer programming based
approach~\cite{roll2004identification}, the Bayesian approach~\cite{juloski2005bayesian}, the bounded error approach~\cite{bemporad2005bounded,ozay2009robust}, sparse optimization approach~\cite{bako2011identification} and the
algebraic approach~\cite{vidal2005generalized}. An underlying assumption in many of these approaches is the assumption that an upper bound on the model order (dimension of continuous state) is known apriori. Furthermore, approaches such as mixed integer linear programming, clustering and SVR do not assume that the switch sequences are observed. Consequently, this causes the identification problem to be NP-hard. In contrast, we assume that the switch sequence is exogenous and known. However, we do not require that an upper bound on the model order be known.

From a system theory perspective, model approximation of SLS has been very well studied \cite{kotsalis2008balanced,birouche2012model,petreczky2013balanced}. These methods mimic balanced truncation--like methods for model reduction and provide error guarantees between the original and reduced system. Despite substantial work on realization theory, identification and model reduction of SLS, there is little work on data driven approaches of SLS identification when model order is unknown, even when the switches are observed.

There is recent interest from the machine learning community in data-driven control and non-asymptotic analysis. In this context, LTI system identification from finite noisy data for the purpose of control has played a central role~\cite{sarkar2019finite,sarkar2018,faradonbeh2017finite,oymak2018non,simchowitz2018learning}. Specifically, \cite{sarkar2019finite} study data driven approaches to learning reduced order approximations of the original model. Furthermore, \cite{sarkar2019finite} studies LTI system identification when model order is unknown which is achieved by constructing low-order approximations of the underlying system from noisy data, where the order grows as the length of data increases. On the other hand,~\cite{gosea2018data} study the Loewner framework for learning state space models directly from data and in principle, allows a trade-off between accuracy and complexity of model learned. However, the authors do not assume the presence of noise in the data. 

\section{Preliminaries}
\label{sec:preliminaries}
Throughout the paper, we will refer to the switches linear system (SLS) with dynamics as Eq.~\eqref{switch_lti} by $M=(C, \{A_i, p_i\}_{i=1}^s, B)$. Furthermore, the SLS $M_1 = (C, \{A_i, p_i\}_{i=1}^s, B)$ is equivalent to $M_2 = (CS^{-1}, \{SA_iS^{-1}, p_i\}_{i=1}^s, SB)$ for any fixed $S$. For a matrix $A$, let $\sigma_i(A)$ be the $i^{\text{th}}$ singular value of $A$ with $\sigma_i(A) \geq \sigma_{i+1}(A)$. Further, $\sigma_{\max}(A) = \sigma_1(A) = \sigma(A)$. Similarly, we define $\rho_i(A) = |\lambda_i(A)|$, where $\lambda_i(A)$ is an eigenvalue of $A$ with $\rho_i(A) \geq \rho_{i+1}(A)$. Again, $\rho_{\max}(A) = \rho_1(A) = \rho(A)$. We define the Kronecker product of two matrices $A, B$ as $A \otimes B$. For any matrix $Z$, define $Z_{i:j, k:l}$ as the submatrix including row $i$ to $j$ and column $k$ to $l$. Further, $Z_{i:j, :}$ is the submatrix including row $i$ to $j$ and all columns and a similar notion exists for $Z_{:, k:l}$. 

We now define stability of a SLS.
\begin{definition}[\cite{costa2006discrete}]
\label{def:stability}
SLS is mean--square stable if $\sum_{i=1}^s p_i A_i \otimes A_i$ is Schur stable, \textit{i.e.}, $\rho\prn*{\sum_{i=1}^s p_i A_i \otimes A_i} < 1$.
\end{definition}
Recall the SLS dynamics in Eq.~\eqref{switch_lti}. Denote by $\tl_i^j = \{\theta_j, \theta_{j-1}, \hdots, \theta_i\} \in [s]^{j-i+1}$ an arbitrary sequence of switches from $i$ to $j$ and
\[
A_{\tl_i^j} \coloneqq A_{\theta_j} A_{\theta_{j-1}} \hdots A_{\theta_i}.
\]
For two switch sequences $\{\theta_2, \theta_1\}, \{\phi_2, \phi_1\}$ define a concatenation operator `$:$' as $\{\theta_2, \theta_1\}:\{\phi_2, \phi_1\}= \{\theta_2, \theta_1, \phi_2, \phi_1\}$. Then $\tl_1^i:\tl_1^{j}$ is concatenation of $\tl_1^{i}, \tl_1^{j}$. It is clear that for any sequence of observed switches $\tl_1^N$, we have the corresponding output $y_N$ of SLS in Eq.~\eqref{switch_lti} as
\begin{align}
y_N &= \underbrace{\sum_{j=2}^{N-1} CA_{\theta_{N-1}}A_{\theta_{N-2}}\hdots A_{\theta_{j}} B u_{j-1} + CB u_{N-1}}_{\text{Input driven terms}} + \underbrace{\sum_{j=2}^{N-1} CA_{\theta_{N-1}}A_{\theta_{N-2}}\hdots A_{\theta_j} \eta_{j-1}  + C\eta_{N-1} + w_N}_{\text{Noise terms}} \label{io_rel}
\end{align}
A measure of distance between two switched linear systems with probabilistic switches is the stochastic $L_2$ gain given by
\begin{definition}[Definition $2.2$ in \cite{kotsalis2008balanced}]
\label{stchastic_l2}
Let the noise $\{\eta_k, w_k\}_{k=1}^{\infty} = 0$. Let $\theta = (\theta_1, \theta_2, \hdots) \in [s]^{\infty}, u = (u_1, u_2, \hdots) \in \Rb^{\infty}$ and $y_{M}^{(\theta, u)}, y_{M_r}^{(\theta, u)} \in \Rb^{\infty}$ be the output sequence, in response to input $u$ and switch sequence $\theta$, of system $M$ and $M_r$ respectively. Then the stochastic $L_2$ distance between $M$ and $M_r$ denoted by $\Delta_{M, M_r}$ is
\[
\Delta^2_{M, M_r} = \sup_{||u||_2 \leq 1}\Ex_{\theta}[||y_M^{(\theta, u)}- y_{M_r}^{(\theta, u)}||^2_2]
\]
\end{definition}
Similar to linear time invariant systems, we will define the doubly infinite system Hankel matrix for SLS. This can be done for example as in~\cite{huang2014minimal}. To summarize, every sequence $\tl_1^N$ has a unique index, \textit{i.e.}, there exists a map $L(\cdot)$ which takes as argument a switch sequence and outputs an integer:
\[
L(\tl_1^N) = \theta_N s^{N-1} + \hdots + \theta_1
\]
with $N = 0 \implies L(\tl_1^N) = 0$. Then the $p(\frac{s^{N+1}-1}{s-1}) \times m(\frac{s^{N+1}-1}{s-1})$ Hankel--like matrix can be constructed as follows: define $r = pL(\tl_{1}^i), s = mL(\tl_{1}^j)$
\begin{align}
[\HN]_{r+1:r+p , s+1:s+m}&= \sqrt{p_{\tl_1^i:\tl_1^j}}CA_{\tl_1^i:\tl_1^j}B    \label{hankel_matrix}
\end{align}
$\forall \hspace{0.5mm}0 \leq i, j \leq N-1$. Then $\Hc^{(\infty)}$ is the doubly infinite system Hankel matrix for SLS. It is a well known fact that $\Hci$ exists and is well defined whenever $p_i A_i \otimes A_i$ is Schur stable. Define $\HN_k$ as 
\begin{align}
[\HN_{k}]_{pL(\tl_{1}^i)+1:pL(\tl_{1}^i) +p , mL(\tl_{1}^j)+1: mL(\tl_{1}^j) + m} = [\HN]_{pL(\tl_{1}^i:\{k\})+1:pL(\tl_{1}^i:\{k\}) +p , mL(\tl_{1}^j)+1: mL(\tl_{1}^j) + m} \label{hck}
\end{align}
\begin{example}
\label{sls_example}
Let $s = 2$. Then $L(\phi) = 0, L(\{1\}) = 1, L(\{2\}) = 2, L(\{1,1\}) = 3, L(\{1,2\}) = 4, \hdots$. As a result 
\begin{align}
\Hci &= \begin{bmatrix}
CB & \sqrt{p_1}CA_1B & \sqrt{p_2} CA_2B & \sqrt{p_1^2}CA_1^2B & \hdots \\
\sqrt{p_1} CA_1B & \sqrt{p_1^2}CA_1^2B  &  \sqrt{p_1 p_2} CA_1A_2B & \sqrt{p_1^3} CA_1^3B & \hdots \\
\sqrt{p_2}CA_2B & \sqrt{p_1p_2}CA_2 A_1 B &   \sqrt{p_2^2} CA_2^2B & \sqrt{p_2p_1^2} CA_2A_1^2B& \hdots \\
\sqrt{p_1^2}CA_1^2B & \sqrt{p_1^3}CA_1^3 B &  \sqrt{p_1^2p_2}CA_1^2A_2B & \sqrt{p_1^4}CA_1^4B & \hdots \\
\sqrt{p_1p_2}CA_1A_2B & \sqrt{p_1^2p_2} CA_1A_2A_1B & \sqrt{p_1p_2^2} CA_1A_2^2B  &  
\sqrt{p_1^3p_2} CA_1A_2A_1^2B & \hdots \\
\vdots & \vdots & \vdots & \vdots & \vdots
\end{bmatrix}\label{matrix1} \\
\Hc_k^{(\infty)}&= \sqrt{p_k}\begin{bmatrix}
CA_k B & \sqrt{p_1}CA_k A_1B & \sqrt{p_2} CA_k A_2B & \hdots \\
\sqrt{p_1} C A_1 A_k B & \sqrt{p_1^2}C A_1 A_k A_1B  &  \sqrt{p_1 p_2} C A_1 A_k A_2B & \hdots \\
\sqrt{p_2}C A_2 A_k B & \sqrt{p_1p_2}C A_2 A_k A_1 B &   \sqrt{p_2^2} CA_2 A_k A_2B &  \hdots \\
\sqrt{p_1^2}CA_1^2 A_k B & \sqrt{p_1^3}CA_1^2 A_k A_1 B &  \sqrt{p_1^2p_2}CA_k A_1^2A_2B & \hdots \\
\sqrt{p_1p_2}C A_1A_2 A_k B & \sqrt{p_1^2p_2} C A_1A_2A_kA_1B & \sqrt{p_1p_2^2} CA_k A_1A_2 A_k A_2B & \hdots \\
\vdots & \vdots & \vdots & \vdots
\end{bmatrix}\label{matrix2}
\end{align}
\end{example}
Note that in the special case of $s = 1$, $p_1 = 1$ and $\Hci$ becomes the usual doubly infinite system Hankel matrix for LTI systems. Furthermore, the order of $M$ is given by $\text{rank}(\Hci)$.
\begin{definition}[\cite{birouche2012model}]
\label{balanced_truncations_definition}
Let $\Hc^{(\infty)} = U \Sigma V^{\top}$ be the doubly infinite system Hankel matrix for SLS $M$, where $\Sigma \in \Rb^{n \times n}$ ($n$ is the model order) and $\Hc_k^{(\infty)}$ be as in Eq.~\eqref{hck}. Then for any $r \leq n$, the $r$--order balanced truncated model parameters are given by 
$$C^{(r)} = [U\Sigma^{1/2}]_{1:p, 1:r}, A^{(r)}_k = \sqrt{p_k}[\Sigma^{-1/2} U^{\top}]_{1:r, :} \Hci_k [V\Sigma^{-1/2}]_{:, 1:r}, B^{(r)} = [\Sigma^{1/2}V^{\top}]_{1:r, 1:m}.$$
For $r > n$, the $r$--order balanced truncated model parameters are the $n$--order truncated model parameters.
\end{definition}
The stability of balanced truncated models is discussed in~\cite{petreczky2013balanced}. We briefly summarize it in Section~\ref{sec:misc_sys}.
\begin{remark}
\label{rem:A_gen}
Note that $A^{(r)}_k$ for $k \in [s]$ can also be obtained in the same fashion as balanced truncated models, specifically $A_r$, for LTI systems in \cite{sarkar2019finite}. We chose this version because it is easier to represent.
\end{remark}
\begin{definition}[\cite{vershynin2019high}]
	\label{subgaussian_rv}
	We say a random vector $v\in\Rb^{d}$ is subgaussian with variance proxy $\tau^{2}$ if 
	$$
	\sup_{||\theta||_2=1}\sup_{p\geq 1}\left\{p^{-1/2}\left(\Ex[\left|\langle v, \theta \rangle\right|^{p}]\right)^{1/p}\right\} = \tau$$ 
	and $\Ex[v] = \textbf{0}$. We denote this by $v\sim\subg(\tau^{2})$.
\end{definition} 


Finally, for two matrices $M_1 \in \Rb^{l_1 \times l_1}, M_2 \in \Rb^{l_2 \times l_2}$ with $l_1 < l_2$, $M_1 - M_2 \triangleq \tilde{M}_1 - M_2$ where $\tilde{M}_1 = \begin{bmatrix}M_1 & 0_{l_1 \times l_2 -l_1} \\
0_{l_2-l_1 \times l_1} & 0_{l_2 - l_1 \times l_2 -l_1}
\end{bmatrix}$. We use non-asymptotic big-oh notation, writing
$f=\bigO(g)$ if there exists a numerical constant such that
$f(x)\leq{}c\cdot{}g(x)$ and $f=\bigOtilde(g)$ if
$f\leq{}c\cdot{}g\max\curly*{\log{}g,1}$.
\section{Contributions}
\label{contributions}
In our work we study the case when noisy $\{y_k, u_k, \theta_k\}_{k=1}^{N}$ is observed and we would like to learn $(C, \{A_i, p_i\}_{i=1}^s, B)$ from observed data when $n$ is unknown. Such a case is relevant when the switches are exogenous but not a control input. The contributions of this paper can be summarized as follows:
\begin{itemize}
    \item We provide the first sample complexity guarantees for SLS identification when the underlying order is unknown. We use tools from finite sample analysis and model reduction theory to derive algorithms that recover system parameters when only finite noisy data is present.  
    \item Central to parameter estimation is estimating the SLS Hankel matrix. A critical step of our algorithm is model selection. We derive a data dependent model selection rule that enables us to construct finite time estimators of the SLS Hankel matrix. The data dependent rule creates a balance between the estimation error and truncation error of the finite time estimator. Specifically, if $N_S$ are the number of samples then we show that 
    \[
    ||\Hci - \hHc^{(\hN)}||_F \leq \wt{\Oc}\prn*{N_S^{-\Delta_s/2}},
    \]
    where $\Hci, \hHc^{(\hN)}$ are the SLS Hankel matrix and finite time estimator respectively, $\Delta_s > 1$ and $\wt{\Oc}(\cdot)$ hides only system dependent constants. This shows that as $N_S \rightarrow \infty$ the finite time estimator is consistent.
    \item Using tools from systems theory, \textit{i.e.} balanced truncation, we obtain parameter estimates by a singular value decomposition (SVD) of the finite time estimator $\hHc^{(\hN)}$. By using a variant of Wedin's theorem we show that the parameter estimates are close to a low order balanced truncation of the true SLS. We show that this a form of subspace recovery where recovery of a lower order approximation depends only on the singular value of the Hankel matrix corresponding to that order and not the lower singular values, \textit{i.e.}, let $M_r = (C^{(r)}, \{A^{(r)}_i\}_{i=1}^s, B^{(r)})$ be the $r$-order balanced truncated approximation of the true SLS and $\wh{M}_r = (\wh{C}^{(r)}, \{\wh{A}^{(r)}_i\}_{i=1}^s, \wh{B}^{(r)})$ be the estimates of our algorithm then 
    \[
    ||M_r - \wh{M}_r|| \leq \wt{\Oc} \prn*{\prn*{r \vee \frac{1}{\sqrt{\sigma_r}}} \cdot{} N_S^{-\Delta_s/2}}
    \]
    where $||\cdot||$ is an appropriately defined norm and $\sigma_i$ are the SLS Hankel singular values.
\end{itemize}

\section{Problem Formulation and Discussion}
\label{sec:problem_formulation}
 \subsection{Data Generation}
 Let $\Mc_n$ be the class of $n$-dimensional mean-squared SLS models, \textit{i.e.}, 
 \[
 \Mc_n = \curly*{(C, \{A_i, p_i\}_{i=1}^s, B) \mid  A_i \in \Rb^{n \times n}, \quad{} \rho\prn*{\sum_{i=1}^s p_i A_i \otimes A_i} < 1}.
 \]
 Assume that the data is generated by $M = (C, \{A_i, p_i\}_{i=1}^s, B) \in \Mc_{n}$ for some unknown $n$. Suppose we observe the noisy output time series $\{y_t \in \mathbb{R}^{p \times 1}\}_{t=1}^T$ and $\{\theta_t \in [s]\}_{t=1}^T$ in response to chosen input series $\{u_t \in \mathbb{R}^{m \times 1}\}_{t=1}^T$. We refer to this data generated by $M$ as $Z_T = \{(u_t, y_t, \theta_t)\}_{t=1}^T$. We enforce the following assumptions on $M$
 \begin{assumption}
 \label{noise_assumption}
 $\{\eta_t, w_t\}_{t=1}^{\infty}$ are i.i.d $\subg(1)$. Furthermore $x_0 = 0$ almost surely. We will only select inputs $\{u_t\}_{t=1}^T$ that are $\Nc(0, I)$. 
 \end{assumption}
 \begin{assumption}
 \label{l2_assumption}
 The number of distinct switches $s$ is known. Furthermore, there exists $\beta > 0$ such that 
 $$\sup_{N \geq 0}{\Big\{||CA_{\tl_1^N}B||^2_F,  ||CA_{\tl_1^N}||^2_F\Big\}}  \leq \beta^2, \quad{} \sup_{i \in [s]}||A_i||_2 \leq \gamma.$$
 \end{assumption}

The goal is to identify $\{C, \{p_i, A_i\}_{i=1}^s, B\}$ from data $\{y_k, u_k, \theta_k\}_{k=1}^{\infty}$ when $n$ (or its upper bound) is unknown. We next describe the problem setup. We assume that the switched linear system can be restarted multiple times as shown in Fig.~\ref{fig:system_setup}. 
\begin{figure}[h]
\centering
    \includegraphics[width=0.4\columnwidth]{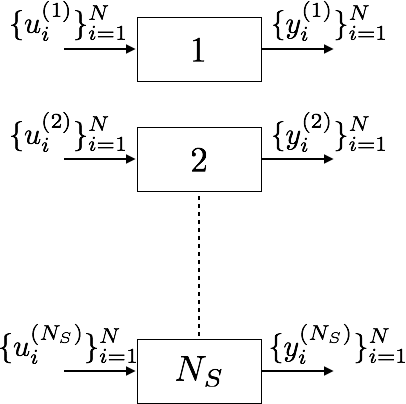}
    \caption{System Setup}
    \label{fig:system_setup}
\end{figure}

Specifically, the data collection process is as follows: each SLS system, or sample, $t \in [N_S]$ is allowed to run for a length of time $N$, also known as the rollout length. We define $N_S$ as the sample complexity. Let $\theta_k^{(t)}, y_k^{(t)}, u_k^{(t)}$ denote the switch, output and input respectively at rollout time $k$ for sample $t$. Now define the set $\Nc_{m_i}$ as 
\begin{equation}
    \label{sample_sequence}
    \Nc_{m_l} \coloneqq \{(t, k)| (\theta^{(t)}_{k+l-1}, \theta^{(t)}_{k+l-2}, \hdots, \theta^{(t)}_{k}) = m_l \in [s]^l\}
\end{equation}
$\Nc_{m_l}$ is the set of occurrences of the switch sequence $m_l$ with $N_{m_l} = |\Nc_{m_l}|$. We briefly describe the intuition behind the setup in Fig.~\ref{fig:system_setup}. 

\subsection{Significance of $N$ and $N_S$}
\label{sec:N_NS}
Our algorithm can be interpreted as an extension to the approach described for LTI systems in~\cite{sarkar2019finite}. However, the main difference between the LTI Hankel matrix and SLS Hankel matrix is that the dimension of the SLS Hankel grows exponentially in the rollout length in contrast to the linear growth for LTI systems. The reason behind this is the following: at any time $i$ any one of $s$ switches may be observed, as a result if a SLS is allowed to run for $N$ steps, the total number of switch sequences possible is $s^N$. To correctly estimate the underlying model, one needs to estimate $CA_{\tl_{1:N}}B$ for each of the sequences $\tl_{1:N} \in [s]^N$.   

At a high level, rolling out the SLS allows us to understand its ``mixing'' properties, \textit{i.e.}, for a given $N$ we observe noisy versions of $CA_{\tl_1^i}B$ for $i \in [N]$. $\{CA_{\tl_1^i}B\}_{i=1}^N$ is then used to construct the SLS Hankel matrix. From this perspective, a large value of $N$, or rollout length, gives us more information about the true SLS.

However, because the observations obtained from a single SLS (or trajectory) are noisy we would like multiple i.i.d copies of the SLS (or trajectories) to ``average'' out the noise and recover $CA_{\tl_1^i}B$. As discussed before for a rollout length of $N$, we need to estimate $s^N$ parameters to construct the Hankel matrix, therefore to ensure that we do not have more parameters than independent samples we need to ensure that $s^N < N_S$. This introduces a trade-off: $N$ cannot be too small as that could lead to high truncation of the Hankel matrix and it should not be too large compared to $N_S$ as that will induce error in estimation of the Hankel matrix parameters. In view of this trade-off, an important step in our SLS Hankel estimation (Algorithm~\ref{alg:regression_estimates}) is that we set $CA_{\tl_1^i}B = 0$ for sequences $\tl_1^i$ that do not occur ``often enough''. To that end, define
\begin{align}
    N_{up} &= \inf\curly*{l \Big| N_{\tl_1^{l}} <  2\prn*{m + \log{\frac{2s_{l}}{\delta}}} \hspace{1.5mm}\forall \tl_1^{l} \in [s]^{l}} \nonumber \\
    s_h &= \frac{s^{h+1} - 1}{s-1}. \label{max_rollout}
\end{align}
$N_{up}$ is the minimum length for which the number of occurrences of \textit{all} switch sequences of that length does not exceed a certain pre-determined threshold. If the SLS is allowed to ``roll-out'' beyond this length we do not have enough samples to accurately learn the parameters corresponding to that length.

\begin{algorithm}[h]
	\caption{Choosing rollout length}
	\label{alg:N_select}
	\textbf{Input} Input: $\{u^{(l)}_m\}_{l=1, m=1}^{l=N_S, m = N_S }$\\
	\textbf{Output} $N_{up}$
	\begin{algorithmic}[1]
	\For{$i= 1, \hdots, N_S$}
	\State Collect output $(y_{i}^{(l)}, \theta_i^{(l)})_{l=1}^{l=N_S}$
	\If{$\exists \tl_1^i$ such that $N_{\tl_1^i} \geq 2\prn*{m + \log{(2s_i/\delta)}}$}
		\State \textbf{continue}
	\Else 
	    \State $N_{up} = i-1$
	    \State \textbf{break}
	\EndIf
\EndFor
	\State \textbf{Return}: $N_{up}$, $\{u^{(l)}_m, y^{(l)}_m,  \theta_{m}^{(l)}\}_{l=1, m=1}^{l=N_S, m = N_{up}}$
	\end{algorithmic}
\end{algorithm}
We next show that $N_{up}$ grows at most logarithmically with the number of samples $N_S$. This is inline with the intuition that when $N_S$ samples are present we should be able to learn $O(N_S)$ parameters.
\begin{proposition}
\label{prop:n_up_noproof}
For $N_{up}$ that is the output of Algorithm~\ref{alg:N_select} we can show with probability at least $1-\delta$ that 
\[
\frac{N_{up}}{\log{N_{up}}} \geq \prn*{\frac{\log{2m} + \log{N_S} + \log{\log{(1/\delta)}}}{\log{\frac{1}{p_{\max}}}}}
\]
where $p_{\max} = \max_{1 \leq i \leq s} p_i$. 
\end{proposition}
The proof of this can be found in Section~\ref{sec:n_up}.
\section{Algorithmic Details}
\label{sec:algorithm_details}

The goal is to learn the infinite Hankel-like matrix $\Hci$, from which we can obtain the system parameters $\sys$ by balanced truncation. However due to the presence of only finite noisy data we will instead construct \textit{finite time} estimators of $\Hci$ from data. We will show that this estimator approaches $\Hci$ as the length of data increases. Furthermore, for finite $N_S$ we can obtain ``good'' lower order approximations of the true SLS from our estimator. Broadly, our algorithmic approach can be encapsulated in Fig~\ref{fig:algo_approach}.
\begin{figure}[H]
\centering
    \includegraphics[width=0.4\columnwidth]{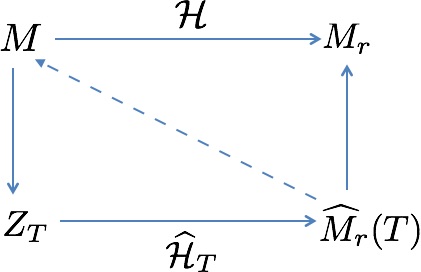}
    \caption{Here $T=N \times N_S$, the total number of observations available. Let $M$ be the true data generating SLS where $Z_T = \{u^{(l)}_m, y^{(l)}_m,  \theta_{m}^{(l)}\}_{l=1, m=1}^{l=N_S, m = N}$ and $M_r$ be a good $r$-order approximation of $M$ obtained by a procedure $\Hc$. Let $\wh{M}_r(T)$ be a $r$-order model estimate obtained by algorithm $\wh{\Hc}_T$ which takes as input $Z_T$. Since $n$ is apriori unknown any parametrization of the model may lead to inconsistent estimators. Instead, for any $T$ we estimate SLS $\wh{M}_r(T)$, that is ``close'' to $M_r$. As $T \rightarrow \infty$, $r \rightarrow n$ and $\wh{M}_r(T)$ approaches $M$.}
    \label{fig:algo_approach}
\end{figure}

The key steps of our algorithm are summarized below. 
\begin{enumerate}
    \item[(a)] Hankel submatrix estimation: Estimating $\Hc^{(l)}$ for every $1 \leq l \leq N$. We refer to these estimators as $\{\wh{\Hc}^{(l)}\}_{l=1}^N$. This is achieved by estimating $\Theta_{\tl_1^i} = CA_{\tl_1^i}B$ for every switch sequence $\tl_{1}^{i} \in [s]^{i}$ where $i \in [N]$.
    \item[(b)] Model selection: From the estimators $\{\wh{\Hc}^{(l)}\}_{l=1}^N$ select $\wh{\Hc}^{(\hN)}$ in a data dependent way such that it ``best'' estimates $\Hci$.
    \item[(c)] Parameter estimation: We do a singular valued decomposition of $\wh{\Hc}^{(\hN)}$ to obtain parameter estimates for a $\hN$-order balanced truncated model.
\end{enumerate}
We summarize some of the key notation below.
\begin{table*}[h]
	\begin{center}
		\begin{tabular}{|l|}
			\hline
			$N_S$: Sample complexity\\
			\hline
			$s_d = \frac{s^{d+1}-1}{s-1}$\\
			\hline
			$N_{up}$: $\inf\curly*{l \Big| N_{\tl_1^{l}} <  \const\prn*{m + \log{\frac{2s_{l}}{\delta}}} \hspace{1.5mm}\forall \tl_1^{l} \in [s]^{l}}$\\
			\hline
			$m$: Input dimension, $p$: Output dimension\\
			\hline
			$\delta$: Error probability \\
			\hline
			$\rho_{max}$: $\rho \prn*{\sum_{i=1}^s p_i A_i \otimes A_i}$\\
			\hline
			$\mu(d) = \sqrt{d}\prn*{d \log{(3s/\delta)} + p\log{\prn*{5 \beta d}} + m}$  \\
			\hline
			$\alpha(d) = \mu(d) \sqrt{2 s_d \cdot{} \frac{d^2}{N_S}}$\\
			\hline
			$\Delta_s = \frac{\log{\prn*{1/\rho_{\max}}}}{\log{\prn*{s/\rho_{\max}}}}$ \\
			\hline
		\end{tabular}
		\caption{Summary of constants} \label{notation}
	\end{center}
\end{table*}
\subsection{Hankel submatrix estimation}
\label{sec:hankel_est}
To estimate the Hankel submatrix we learn each of its entries by linear regression, \textit{i.e.}, for every switch sequence $\tl_{1}^{i} \in [s]^{i}$, we want to learn $\Theta_{\tl_1^i} = CA_{\tl_1^i}B$ and its probability $p_{\tl_1^i}$ of occurrence. This is shown in Algorithm~\ref{alg:regression_estimates}. 
\begin{algorithm}[H]
	\caption{Regression Estimates}
	\label{alg:regression_estimates}
	\textbf{Input} Data: $\{u^{(l)}_m, y^{(l)}_m,  \theta_{m}^{(l)}\}_{l=1, m=1}^{l=N_S, m = N_{up}}$\\
	\textbf{Output} Estimates: $\wh{\Hc}^{(N)}, \{\wh{p}_i\}_{i=1}^s$
	\begin{algorithmic}[1]
	\State $N = N_{up}$
	\For{$i= 1, \hdots, N$}
	\If{$N_{\tl_1^i} \geq 2(m + N\log{\frac{2s}{\delta}})$}
		\State $\wh{\Theta}_{\tl_1^i} = \arg \inf_{\Theta} \sum_{(t, k) \in \Nc_{\tl_1^i}}||y_k^{(t)} - \Theta u_k^{(t)}||_F^2$
	\Else 
	    \State $\wh{\Theta}_{\tl_1^i} = 0$
	\EndIf
	\State 
	\[
    \wh{p}_{\tl_1^i} =
    \begin{cases}
        \frac{N_{\tl_1^i}}{N_S (N-i+1)} ,& \text{if } i > 0\\
        1,              & \text{otherwise}
    \end{cases}
    \]
    \For{$a= 0,1, \hdots, N$}
    \State $b = N-a$
    \State 	\begin{align}
        [\hHN]_{pL(\tl_{1}^a)+1:pL(\tl_{1}^a)+p , mL(\tl_{1}^b) + 1: mL(\tl_{1}^b) + m} = \sqrt{\hp_{\tl_1^a:\tl_1^b}}\wh{\Theta}_{\tl_1^a:\tl_1^b}   \label{est_hankel_matrix}
    \end{align} 
    \EndFor
\EndFor
	\State \textbf{Return}: $\hHN, \{\wh{p}_i\}_{i=1}^s$.

	\end{algorithmic}
\end{algorithm}

Our next result  bounds the error rates obtained from the regression. The proof of this follows standard analysis in statistical learning literature such as~\cite{tu2017non}.
\begin{proposition}
\label{regression_result}
Let $N$ be the rollout length. Fix $\delta > 0$ and sequence $\tl_1^i \in [s]^{i}$. Let $\wh{\Theta}_{\tl_1^i}$ be the following solution
\[
\wh{\Theta}_{\tl_1^i} = \arg \inf_{\Theta} \sum_{(t, k) \in \Nc_{\tl_1^i}}||y_k^{(t)} - \Theta u_k^{(t)}||_F^2
\]
where $\{u^{(t)}_k\}_{t, k=1}^{\infty}$ are i.i.d isotropic Gaussian (or isotropic $\subg(1)$) random variables. Then whenever $\sqrt{N_{\tl_1^i}} \geq c \prn*{\sqrt{m}  + \sqrt{N\log{(6s/\delta)}}}$ we have with probability at least $1 -\delta$ that
\begin{equation}
||CA_{\tl_1^i}B - \wh{\Theta}_{\tl_1^i}||_F \leq  10\min{(\sqrt{p}, \sqrt{m})} \cdot{} \beta \sqrt{\frac{\mu^{2}(N)}{N_{\tl_1^i}} }.    \label{error}
\end{equation}
Here $c$ is an absolute constant. Furthermore, $\Ex[N_{\tl_1^i}] = p_{\tl_1^i}N_S (N-i+1)$.
\end{proposition}
\begin{proof}
The proof of the first part can be found in the appendix as Proposition~\ref{prop:reg_result}. The details are standard in statistical learning theory and require applications of Bernstein-type inequalities. To see that $\wh{p}_{\tl_1^i:\tl_1^j}$ is an unbiased estimator for $p_{\tl_1^i:\tl_1^j}$, recall the experiment set up: we run $N_S$ identical samples of the SLS for length $N$. Then for each sample $i \leq N_S$, any sequence $\tl_1^k$ can start at position $1, 2, \hdots, N-k+1$. Thus for the sample $i$ the number of occurrences of $\tl_1^k$ is given by $\sum_{l=1}^{N-k+1} \textbf{1}^{(i)}_{\{\tl_1^k \text{ starts at position }l\}}$, then $N_{\tl_1^k}$ is given by
\begin{equation}
    N_{\tl_1^k} = \sum_{i=1}^{N_S} \sum_{l=1}^{N-k+1} \textbf{1}^{(i)}_{\{\tl_1^k \text{ starts at position }l\}}
\end{equation}
and by linearity of expectation it is clear that $\Ex[N_{\tl_1^k}] = p_{\tl_1^k}N_S(N-k+1)$.
\end{proof}
It is clear from Proposition~\ref{regression_result} that $\wh{p}_{\tl_1^i}$ is an unbiased estimator of $p_{\tl_1^i}$ because 
\[
\Ex[\wh{p}_{\tl_1^i}] = \frac{\Ex[N_{\tl_1^i}]}{N_S(N-i+1)} = p_{\tl_1^i}
\]

An important feature of this algorithm is that when we have less data for a certain switch sequence $\tl_1^i$, \textit{i.e.}, $N_{\tl_1^{i}} < 2(m + N\log{\frac{2s}{\delta}})$ we simply set the parameter corresponding to the switch sequence to zero. The reason being that when the occurrences of a particular switch sequence is low, regression cannot be used to obtain reliable estimates. On the other hand, setting these parameters to zero does not lead to a high estimation error, which we show in Theorem~\ref{thm1}. The reason is that in the Hankel matrix each entry $CA_{\tl_1^i}B$ is scaled by its probability of occurrence $\sqrt{p_{\tl_1^i}}$, as a result for low probability sequences the estimation error also remains low.

Then we have the following estimation error upper bound,
\begin{theorem}
\label{thm1}
Fix $\delta > 0$ and $N$. We have with probability at least $1-\delta$, 
\[
||{\hHc}^{(N)}-\Hc^{(N)}||_F^2 \leq 2\beta^2 \prn*{\frac{N^2}{N_S}}s_N \cdot{} \mu^2(N) = \beta^2 \alpha^2(d) .
\]
Here $s_{k} =\frac{s^{k+1}-1}{s-1}$.
\end{theorem}
The proof of this result can be found in Proposition~\ref{prop:estimation_error} in the appendix. Theorem~\ref{thm1} provides a finite time error bound in estimating $\Hc^{(N)}$ using Algorithm~\ref{alg:regression_estimates} for any fixed $N$. The bound depends on quantities that are known apriori. This will be critical in designing a data dependent rule for model selection.

\subsection{Model Selection}
\label{sec:model_selection}
At a high level, we want to choose $\hHc^{(\hN)}$ from $\{\hHc^{(l)}\}_{l=1}^{N_{up}}$ such that $\hHc^{(\hN)}$ is a good estimator of $\Hc^{(\infty)}$. Our idea of model selection is motivated by~\cite{goldenshluger1998nonparametric}. For any $\hHc^{(l)}$, the error from $\Hci$ can be broken as:
\[
||\hHc^{(l)} - \Hci||_F \leq \underbrace{||\hHc^{(l)} - \Hc^{(l)}||_F}_{=\text{Estimation Error}}  + \underbrace{||\Hc^{(l)} - \Hci||_F}_{=\text{Truncation Error}}.
\]
The size of $\hHc^{(l)}$ is made equal to $\Hci$ by padding it with zeros. We would like to select a $l=\hN$ such that it balances the truncation and estimation error in the following way:
\begin{equation}
c_2 \cdot \text{Upper bound } \geq c_1 \cdot \text{Estimation Error} \geq \text{Truncation Error}    \label{eq:balancing_1}
\end{equation}
where $c_i$ are absolute constants. Such a balancing ensures that 
\begin{equation}
||\hHc^{(l)} - \Hci||_F \leq c_2 \cdot (1/c_1 + 1) \cdot \text{Upper bound }. \label{eq:balancing}
\end{equation}
Note that such a balancing is possible because the estimation error increases as $l$ grows and truncation error decreases with $l$. Furthermore, a data dependent upper bound for estimation error can be obtained from Theorem~\ref{thm1}. Unfortunately $(C, \{A_i\}_{i=1}^s, B)$ are unknown and it is not immediately clear on how to obtain such a bound for truncation error. 

To achieve this, we first define a truncation error proxy, \textit{i.e.}, how much do we truncate if a specific $\hHc^{(l)}$ is used. For a given $l$, we look at $||\hHc^{(l)} - \hHc^{(d)}||_F$ for $N_{up} \geq d \geq l$. This measures the additional error incurred if we choose $\hHc^{(l)}$ as an estimator for $\Hc^{(\infty)}$ instead of $\hHc^{(d)}$ for $d > l$. Then we pick $\hN$ as follows:
\begin{equation}
\hN \coloneqq  \inf\Bigg\{l \Bigg| ||\hHc^{(d)} - \hHc^{(l)}||_F \leq \beta(\alpha(d) + 2 \alpha(l)) \quad{} \forall N_{up} \geq d \geq l\Bigg\}. \label{eq:hd_eq}    
\end{equation}
where $\alpha(d)$ is from Table~\ref{notation}. A key step will be to show that for any $d \geq l$, whenever 
\[
||\hHc^{(d)} - \hHc^{(l)}||_F \leq c \beta \cdot \alpha(d)
\]
ensures that  
\[
||\hHc^{(l)} - \Hc^{(\infty)}||_F \leq c \beta \cdot  \alpha(d) \quad{} \text{and} \quad{} ||\hHc^{(d)} - \Hc^{(\infty)}||_F \leq c \beta \cdot \alpha(d) 
\]
and there is no gain in choosing a larger Hankel submatrix estimate. By picking the smallest $l$ for which such a property holds for all larger Hankel submatrices, we ensure that a regularized model is estimated that ``agrees'' with the data. The model selection algorithm is summarized in Algorithm~\ref{alg:d_choice}.

\begin{algorithm}[H]
	\caption{Choice of $\hN$}
	\label{alg:d_choice}
	\textbf{Output} $\hHc^{(\hN)}, \quad{} \hN$
	\begin{algorithmic}[1]
		\State $s_d, \alpha(d)$ from Table~\ref{notation}.
		\State $\hN =\inf \curly*{l \Big|  ||\hHc^{(d)} - \hHc^{(l)}||_F \leq \beta \prn*{2\alpha(d) + \alpha(l)} }$. \\
		\Return $\hHc^{(\hN)}, \quad{} \hN$
	\end{algorithmic}
\end{algorithm}
First we show that $\hN$, \textit{i.e.}, the output of Algorithm~\ref{alg:d_choice} does not grow strongly with $N_S$. Using this fact, we will show that $\hHc^{(\hN)}$ is a good finite time estimator of $\Hci$.

\begin{theorem}
\label{thm:finite_time}
Assume that 
\[
N_S \geq \prn*{\frac{N_{up}\cdot{} c(n) \rho_{\max}}{(1-\rho_{\max}) \cdot{}\log{N_S}}}^{\frac{\log{\prn*{1/p_{\max}}}}{\log{\prn*{1/\rho_{\max}}}}}
\]
where $c(n)$ is a constant that depends on $m, n, p$ only. Then we have with probability at least $1 - \delta$ that
\[
||\hHc^{(\hN)} - \Hci||^2_F = \wt{\Oc}(N_S^{-\Delta_s}).
\]
\end{theorem}
\begin{proof}
The proof details are given in Proposition~\ref{prop:ns_hn}. We sketch the proof here. First step is to show that there exists $\Ns$ such that, for $N_S$ that satisfies the conditions of the Theorem, we have $$\text{Truncation Error} \leq \text{Estimation Error}$$ 
and then showing that the estimation error decays with $N_S$ (Proposition~\ref{prop:n_star}). The final step is to show that $\hN \leq \Ns$, and using subadditivity to show that $||\hHc^{(\hN)} - \Hci||_F$ also decays with $N_S$.
\end{proof}

Theorem~\ref{thm:finite_time} gives us an error bound on the estimation of $\Hci$ that decays to zero as the number of samples increase via a method of model selection from data dependent quantities only. In the following section we will use this data dependent error bound to construct parametric estimates. Furthermore, note the special case of $s=1$ (LTI systems). We find that $\Delta_s=1$ when $s=1$, \textit{i.e.}, the error rate is $\wt{\Oc}(N_S^{-1})$. This is exactly equal to the error rate obtained for LTI systems (See Theorem 5.2 in~\cite{sarkar2019finite}).

\subsection{Parameter Estimation}
\label{sec:param_est}
Finally, to obtain the system parameters of the SLS we use the following balanced truncation based algorithm. This involves a SVD of the Hankel matrix and the error bounds can then be obtained by using a variant of Wedin's theorem.

\begin{algorithm}[H]
	\caption{Learning SLS parameters}
	\label{alg:learn_sls}
	\textbf{Input} $\hHc^{(\hN)}$: Hankel Estimator, $\{\wh{p}_i\}_{i=1}^s$: Probability estimates,  $r$: Order\\ 
	\textbf{Output} System Parameters $(\tC, \{\tA_i\}_{i=1}^s, \tB)$
	\begin{algorithmic}[1]
		\State Pad $\hHc^{(\hN)}$ with zeros to make it of size $4ps_{\hN} \times 4ms_{\hN}$
		\State $\hHc^{(\wh{N})}_{1:4ps_{\hN} - p, 1:4ms_{\hN} - m} = \hU \hSigma \hV^{\top}$ and \\
		\State $\wh{C}^{(r)} = [\hU\hSigma^{1/2}]_{1:p, 1:r}, \wh{B}^{(r)} = [\hSigma^{1/2}\hV^{\top}]_{1:r, 1:m}$
		\State $\wh{A}_i^{(r)} = \wh{p}_i^{-1/2} [\hSigma^{-1/2} \hU^{\top}]_{1:r, :} \hHc^{({N})}_{i} [\hV \hSigma^{-1/2}]_{:, 1:r}$
		\State \textbf{Return}: $(\wh{C}^{(r)}, \{\wh{A}_i^{(r)}\}_{i=1}^s, \wh{B}^{(r)})$ 
	\end{algorithmic}
\end{algorithm}
By Theorem~\ref{thm:finite_time} we know that $\wh{\Hc}^{(\hN)}$ is a ``good'' finite time estimator for $\Hci$, \textit{i.e.}, the error decays to zero as $N_S$ goes to infinity. Since all lower order approximations of mean square stable SLS can be obtained from $\Hci$, by using subspace perturbation results we will argue that these approximations can be estimated from $\wh{\Hc}^{(\hN)}$. To state the main result we define a quantity that measures the singular value weighted subspace gap of a matrix $S$:
\[
\Gamma(S, \epsilon) = \sqrt{ {\sigma}^1_{\max}/\zeta_{1}^2 + {\sigma}^2_{\max}/\zeta_{2}^2 + \hdots + {\sigma}^{l}_{\max}/\zeta_{l}^2},
\]
where $S = U \Sigma V^{\top}$ and ${\Sigma}$ is arranged into blocks of singular values such that in each block $i$ we have $\sup_j \sigma^{i}_{j} - \sigma^{i}_{j+1} \leq \epsilon$, \textit{i.e.}, 
\[
{\Sigma} = \begin{bmatrix}
\Lambda_1 & 0 & \ldots & 0 \\
0 & \Lambda_2 & \ldots & 0 \\
\vdots & \vdots & \ddots & 0 \\
0 & 0 & \ldots & \Lambda_l \\
\end{bmatrix}
\]
where $\Lambda_i$ are diagonal matrices, $\sigma^{i}_{j}$ is the $j^{th}$ singular value in the block $\Lambda_i$ and $\sigma^{i}_{\min}, \sigma^{i}_{\max}$ are the minimum and maximum singular values of block $i$ respectively. Furthermore,
\[
\zeta_{i} = \min{({\sigma}^{i-1}_{\min}-{\sigma}^{i}_{\max}, {\sigma}^{i}_{\min}-{\sigma}^{i+1}_{\max})}
\]
for $1 < i < l$, $\zeta_{1} ={\sigma}^{1}_{\min}-{\sigma}^{2}_{\max}$ and $\zeta_{l} = \min{({\sigma}^{l-1}_{\min}-{\sigma}^{l}_{\max}, {\sigma}^{l}_{\min})}$. Informally, the $\zeta_i$ measure the singular value gaps between each blocks. It should be noted that $l$, the number of separated blocks, is a function of $\epsilon$ itself. For example: if $\epsilon = 0$ then the number of blocks correspond to the number of distinct singular values. On the other hand, if $\epsilon$ is very large then $l = 1$. 

Let $\Hci = U \Sigma V^{\top}$ where $\Sigma \in \Rb^{n \times n}$ because SLS is finite rank. Define $\Gamma_0 \coloneqq \sup_{\tau \geq 0} \Gamma(\Sigma, \tau) < \infty$.
\begin{theorem}
	\label{balanced_truncation}
Let $N_S$ satisfy the conditions of Theorem~\ref{thm:finite_time} and $\hd$ be the number of non-zero singular values of $\hHc^{(\hN)}$, $M$ be the true unknown model and $(C^{(r)}, \{A^{(r)}_i\}_{i=1}^s, B^{(r)})$ be its $r$-order balanced truncated parameters. 

Then, for $r \leq \hd$, there exists an orthogonal transformation $Q$ such that with probability at least $1-\delta$ we have 
\begin{enumerate}
    \item[(a)] If $\sigma_r(\Hci) = \Omega(N_S^{-\Delta_s/2})$ then 
    \begin{align*}
	    \max{ \prn*{||\wh{C}^{(r)} - C^{(r)}Q||_2, ||\wh{B}^{(r)} - Q^{\top}B^{(r)}||_2}} &= \wt{\Oc}\prn*{\frac{\Gamma_0 \cdot{} N_S^{-\Delta_s/2}}{\sqrt{\sigma_r}} \vee r N_S^{-\Delta_s/2}} ,\\
    	||Q^{\top}A^{(r)}Q - \wh{A}^{(r)}||_2 &= \wt{\Oc}\prn*{\frac{\gamma \Gamma_0 \cdot{} N_S^{-\Delta_s/2}}{\sqrt{\sigma_r}} \vee r N_S^{-\Delta_s/2}}.
	\end{align*}
	\item[(b)] If $\sigma_r(\Hci) = o(N_S^{-\Delta_s/2})$ then 
	\begin{align*}
	    \max{ \prn*{||\wh{C}^{(r)} - C^{(r)}Q||_2, ||\wh{B}^{(r)} - Q^{\top}B^{(r)}||_2}} &= \wt{\Oc}\prn*{\frac{\Gamma_0 \cdot{} N_S^{-\Delta_s/2}}{\sqrt{\sigma_r}} \vee \sqrt{r N_S^{-\Delta_s/2}}} ,\\
	    ||Q^{\top}A^{(r)}Q - \wh{A}^{(r)}||_2 &=  \wt{\Oc}\prn*{\frac{\gamma \Gamma_0 \cdot{} N_S^{-\Delta_s/2}}{\sqrt{\sigma_r}} \vee \sqrt{r N_S^{-\Delta_s/2}}}.
	\end{align*}
\end{enumerate}
\end{theorem}
The proof of Theorem~\ref{balanced_truncation} can be found in Section~\ref{sec:thm_balanced_proof} in the appendix. We note that when $\sigma_r$ is large then estimation error for the lower order approximation decays as $\wt{\Oc}(N_S^{-\Delta_s/2})$. On the other hand, if $\sigma_r$ is below a certain threshold then the error decays at a slower rate of $\wt{\Oc}(N_S^{-\Delta_s/4})$.
\section{Discussion}
\label{discussion}
In this work we provide finite sample error guarantees for learning realizations of SLS when stability radius or order is unknown. Specifically, we construct a finite dimensional Hankel--like matrix from finite noisy data that estimates $\Hci$ and obtain parameter estimates by balanced truncation. The two key features of our algorithm are choosing the rollout length and the size of the Hankel--like matrix using data dependent methods. The data dependent methods carefully balance the estimation and truncation errors incurred. Under stated assumptions, we obtain $O\prn*{\sqrt{N_S^{-\Delta_s}}}$ finite time error rates which coincide with the optimal rates for LTI systems when $s=1$. 

Due to the nature of the analysis we believe that the results can be easily extended to the case when $\{\theta_t\}_{t=1}^{\infty}$ evolution is more complex, for e.g.: state dependent or a markov chain. Furthermore, we assumed in this paper that the discrete switches are completely observable. A potential future direction could be to combine these results with clustering based methods described in~\cite{ferrari2003clustering} to build a more general framework of SLS identification.
\bibliographystyle{IEEEtran}
\bibliography{bibliography}{}
\newpage
\appendix
\section{Probabilistic Inequalities}
\label{sec:prob_ineq}
\begin{proposition}[\cite{vershynin2010introduction}]
	\label{eps_net}
	Let $M$ be a random matrix. Then we have for any $\epsilon < 1$ and any $w \in \Sc^{d-1}$ that 
	\[
	\Pb\prn*{||M|| > z} \leq (1 + 2/\epsilon)^d \Pb\prn*{||Mw|| > (1-\epsilon)z}.
	\]
\end{proposition}
\begin{proposition}[\cite{vershynin2019high} Theorem 5.39]
\label{isometry}
	Let $U$ be a $T \times d$ matrix whose rows,
	$\{u_i\}_{i=1}^T$, are independent and isotropic random
	vectors in $\Rb^{d}$, belonging to $\subg(\tau^2)$. Then whenever 
	\[
	\sqrt{T} \geq c \tau^2\prn*{\sqrt{d}  + \sqrt{\log{(2/\delta)}}},
	\]
        we have that with probability at least $1- \delta$,
	\[
          \frac{3}{4}\cdot{}  I \preceq \sum_{i=1}^T u_i u_i^{\top} \preceq \frac{5}{4} \cdot{} I,
	\]
	where $c>0$ is a numerical constant.
\end{proposition}
\begin{theorem}[Theorem 8.5 in~\cite{sarkar2019finite}]
	\label{selfnorm_main}
Let $\{\bcF_t\}_{t=0}^{\infty}$ be a filtration. Let $\{\eta_{t} \in \Rb^m, X_t \in \Rb^d\}_{t=1}^{\infty}$ be stochastic processes such that $\eta_t, X_t$ are $\bcF_t$ measurable and $\eta_t$ is $\bcF_{t-1}$-conditionally $\subg(\tau^2)$ for some $L > 0$.
	For any $t \geq 0$, define 
	$
	V_t = \sum_{s=1}^t X_s X_s^{\top}, S_t = \sum_{s=1}^t  X_s\eta_{s+1}^{\top}
	$.
	Then for any $\delta > 0, V \succ 0$ and all $t \geq 0$ we have with probability at least $1-\delta$ 
	\[
	S_t^{\top}(V + V_t)^{-1}S_t  \leq  2\tau^2 \prn*{\log{\frac{1}{\delta}} + \log{\frac{\text{det}(V+V_t)}{\text{det}(V)}} + m}.
	\]
\end{theorem}
\begin{theorem}[Theorem 2.1~\cite{rudelson2013hanson}]
    \label{subg-conc}
    Let $M$ be a fixed matrix. Consider a random vector $X$ such that $\Ex[X_i] = 0$, $\Ex[X_i^2] = 1$ and $X \sim \subg(L^2)$. Then for any $t \geq 0$, we have 
    \[
    \Pb(|||AX||_2 - ||A||_F| \geq t) \leq 2 \cdot{} \exp{\prn*{-\frac{ct^2}{L^4 ||A||^2}}}.
    \]
\end{theorem}
\begin{proposition}[Bernstein's Inequality]
\label{bernstein}
Let $\{X_i\}_{i=1}$ be zero mean random variables. Suppose that $|X_i| \leq M$ almost surely, for all $i$. Then, for all positive $t$, 
\begin{equation}
   \Pb(|\sum_{i=1}^n X_i| > t)  \leq \exp{\prn*{-\frac{\frac{1}{2}t^2}{\sum \Ex[X_j^2] + \frac{1}{3}Mt}}}.\label{bern}
\end{equation}
\end{proposition}
Recall that 
$$\hp_{\tl_1^k} = \frac{N_{\tl_1^k}}{N_S(N-k+1)} = \frac{1}{N_S}\sum_{i=1}^{N_S}\underbrace{ \sum_{l=1}^{N-k+1}\frac{\textbf{1}^{(i)}_{\{\tl_1^k \text{starts at }l\}}}{(N-k+1)}}_{x_i}$$
where $0 \leq x_i \leq 1$ and $\{x_i\}_{i=1}^{N_S}$ are i.i.d random variables. Since $\Ex[x_i] = p_{\tl_1^k}$. We can use Bernstein's inequality on $\sum_{i=1}^{N_S} (x_i - p_{\tl_1^k})$. Note $\Ex[(x_i-p_{\tl_1^k})^2] \leq p_{\tl_1^k} $. Then by Bernstein's inequality we have
\begin{proposition}
\label{hoeffding}
Fix $\delta, N > 0$. For all sequences $\tl_{1}^k$ with $k \leq N$ we have simultaneously with probability at least $1-\delta$
\[
\bl \sum_{i=1}^{N_S}x_i - p_{\tl_1^k}N_S  \bl \leq 
\begin{cases}
      \log{\frac{s_N}{\delta}} ,& \text{if }  \log{\frac{s_N}{\delta}}  > p_{\tl_1^k}N_S \\
     \sqrt{p_{\tl_1^k}N_S  \log{\frac{s_N}{\delta}}},              & \text{otherwise}
\end{cases}
\]
where $s_N = \frac{s^{N+1}-1}{s-1}$.
\end{proposition}
\begin{proof}
First, we will show this high probability bound for a specific sequence $\tl_1^k$ and then we will give the bound for all sequences of length at most $N$ by using a union bound. For a specific sequence $\tl_1^k$ we have with probability at least $1 - \delta$, 
\[
\bl \sum_{i=1}^{N_S}x_i - p_{\tl_1^k}N_S  \bl \leq 
\begin{cases}
     \log{\frac{1}{\delta}} ,& \text{if }  \log{\frac{1}{\delta}}  > p_{\tl_1^k}N_S \\
     \sqrt{p_{\tl_1^k}N_S  \log{\frac{1}{\delta}}},              & \text{otherwise}
\end{cases}
\]
This is obtained from Proposition~\ref{bernstein} where $M =  1$ and $\Ex[X_j^2] \leq p_{\tl_1^k}$. Now, the total number of sequences is $s_N = \frac{s^{N+1} - 1}{s-1}$. If we want, the preceding relation to hold simultaneously for all sequences then we have with probability at least $1 - s_N \delta$ that 
\[
\bl \sum_{i=1}^{N_S}x_i - p_{\tl_1^k}N_S  \bl \leq 
\begin{cases}
      \log{\frac{1}{\delta}} ,& \text{if }  \log{\frac{1}{\delta}}  > p_{\tl_1^k}N_S \\
     \sqrt{p_{\tl_1^k}N_S  \log{\frac{1}{\delta}}},              & \text{otherwise}
\end{cases}
\]
for every sequence of length $k \leq N$. By substituting $\delta \rightarrow s_N^{-1} \delta$, we get the desired the bound.
\end{proof}

\section{Miscellaneous Results on Systems Theory}
\label{sec:misc_sys}
\begin{proposition}[Similarity Equivalence]
Given the Hankel matrix $\Hci$. If $\sys$ generates $\Hci$ then $(CS, \{S^{-1}A_iS, p_i\}, S^{-1}B)$ generates $\Hci$ as well, \textit{i.e.}, $\sys$ is equivalent to $(CS, \{S^{-1}A_iS, p_i\}, S^{-1}B)$ in input-output behavior.
\end{proposition}
\begin{proof}
The proof follows from straightforward computation.
\end{proof}

\begin{proposition}[Stability of balanced truncation]
Given a SLS $M = \sys$ and assume that the switches are independent. Let $M_r$ be the $r$-order balanced truncated model, then $M_r$ is mean-squared stable and 
\[
\Delta_{M, M_r} \leq 2 s \sum_{i=r+1}^n \sigma_i
\]
where $\sigma_i$ is the $i^{th}$ singular value of $\Hci$.
\end{proposition}
\begin{proof}
Define $(\tC_i, \tA_i, \tB_i) = (\sqrt{p}_i C, \sqrt{p}_i A_i, \sqrt{p}_i B)$, then $\rho(\sum_{i=1}^s \tA_i \otimes \tA_i) < 1$, \textit{i.e.}, the system with parameters
$$M = \{(\sqrt{p}_i C, \sqrt{p}_i A_i, \sqrt{p}_i B)\}_{i=1}^s$$ 
is strongly stable as in Definition 14 in~\cite{petreczky2013balanced}. Following that we use Lemma 12 and Theorem 6 in~\cite{petreczky2013balanced} to obtain stability and performance guarantees for the reduced system from~\cite{birouche2012model}.
\end{proof}

\begin{proposition}[Proposition 14.2 in~\cite{sarkar2019finite}]
	\label{prop:c_select}
	Let $\Hci = U \Sigma V^{\top}, \wh{\Hc}^{(\hN)} = \wh{U} \wh{\Sigma} \wh{V}^{\top}$ and 
	$$||\Hci - \wh{\Hc}^{(\hN)}||_F \leq \epsilon.$$
	Define the $r$--order balanced truncated parameters
    $$C^{(r)} \coloneqq [U\Sigma^{1/2}]_{1:p, 1:r}, A^{(r)}_k \coloneqq \sqrt{p_k}[\Sigma^{-1/2} U^{\top}]_{1:r, :} \Hci_k [V\Sigma^{-1/2}]_{:, 1:r}, B^{(r)} \coloneqq [\Sigma^{1/2}V^{\top}]_{1:r, 1:m},$$
	and analogously $\wh{C}^{(r)}, \wh{A}^{(r)}_k, B^{(r)}$ for $k \in [s]$ using $\wh{U}, \wh{\Sigma}, \wh{V}$.
	Let $\Sigma$ be arranged into blocks of singular values such that in each block $i$ we have 
	\[
	\sup_j \sigma^{i}_{j} - \sigma^{i}_{j+1} \leq \chi \epsilon
	\]
	for some $\chi \geq 2$, \textit{i.e.}, 
	\[
	\Sigma = \begin{bmatrix}
	\Lambda_1 & 0 & \ldots & 0 \\
	0 & \Lambda_2 & \ldots & 0 \\
	\vdots & \vdots & \ddots & 0 \\
	0 & 0 & \ldots & \Lambda_l \\
	\end{bmatrix}
	\]
	where $\Lambda_i$ are diagonal matrices and $\sigma^{i}_{j}$ is the $j^{th}$ singular value in the block $\Lambda_i$. Then there exists an orthogonal transformation, $Q$, such that 
	\begin{align*}
	\max{ \prn*{||\wh{C}^{(r)} - C^{(r)}Q||_2, ||\wh{B}^{(r)} - Q^{\top}B^{(r)}||_2}} &\leq 2 \epsilon \sqrt{ {\sigma}_{1}/\zeta_{n_1}^2 + {\sigma}_{n_1+1}/\zeta_{n_2}^2 + \hdots + {\sigma}_{\sum_{i=1}^{l-1}n_i+1}/\zeta_{n_l}^2}  \\
	&+ 2\sup_{1 \leq i \leq l}\sqrt{\sigma^i_{\max}} - \sqrt{\sigma^i_{\min}} + \frac{\epsilon}{\sqrt{\sigma_{r}}} \wedge \sqrt{\epsilon} = \zeta,\\
	||Q^{\top}A^{(r)}Q - \wh{A}^{(r)}||_2 &\leq  4 \gamma \cdot \zeta / \sqrt{\sigma_{r}}.
	\end{align*}
	Here $\sup_{1\leq i \leq l}\sqrt{\sigma^i_{\max}} - \sqrt{\sigma^i_{\min}} \leq \frac{\chi }{\sqrt{\sigma^i_{\max}}} \epsilon r \wedge \sqrt{\chi r \epsilon}$ and 
	\[
	\zeta_{n_i} = \min{({\sigma}^{n_{i-1}}_{\min}-{\sigma}^{n_i}_{\max}, {\sigma}^{n_i}_{\min}-{\sigma}^{n_{i+1}}_{\max})}
	\]
	for $1 < i < l$, $\zeta_{n_1} ={\sigma}^{n_1}_{\min}-{\sigma}^{n_2}_{\max}$ and $\zeta_{n_l} = \min{({\sigma}^{n_{l-1}}_{\min}-{\sigma}^{n_l}_{\max}, {\sigma}^{n_l}_{\min})}$.
\end{proposition}

\section{Proofs from the main paper}
\label{appendix:proofs}
\subsection{Proof of Proposition~\ref{prop:n_up_noproof}}
\label{sec:n_up}
\begin{proposition}
\label{prop:n_up}
For $N_{up}$ defined in Eq.~\eqref{max_rollout} we can show with probability at least $1-\delta$ that 
\[
\frac{N_{up}}{\log{N_{up}}} \geq \prn*{\frac{\log{2m} + \log{N_S} + \log{\log{(1/\delta)}}}{\log{\frac{1}{p_{\max}}}}}
\]
where $p_{\max} = \max_{1 \leq i \leq s} p_i$. 
\end{proposition}
\begin{proof}
By definition $N_{\text{up}}$ is the least $l$ for which we have
\[
N_{\tl_1^{l}} <  2\prn*{m + \log{\frac{2s_{l}}{\delta}}} \hspace{2mm} \forall \tl_1^{l} \in [s]^{l}
\]
where $s_l = \frac{s^{l+1}-1}{s-1}$. Note from the first inequality of Proposition~\ref{hoeffding} we have that for all sequences of length $k$, if $N_{\tl_1^k} \leq 2 \log{\frac{s^{k+1}-1}{(s-1)\delta}}$, then it implies that with probability at least $1-\delta$ we get 
\[
2\prn*{ m +  \log{\Big(\frac{s^{N_{up}+1}-1}{(s-1)\delta}\Big)} } \geq p_{\max}^{N_{up}} \cdot{} N_S,
\]
which gives us that 
\[
\frac{N_{up}}{\log{N_{up}}} \geq \prn*{\frac{\log{2m} + \log{N_S} + \log{\log{(1/\delta)}}}{\log{\frac{1}{p_{\max}}}}}.
\]
\end{proof}

\subsection{Proof of Proposition~\ref{regression_result}}
\label{sec:reg_result}
\begin{proposition}
\label{prop:reg_result}
Fix $\delta > 0$ and sequence $\tl_1^i \in [s]^{i}$. Let $\wh{\Theta}_{\tl_1^i}$ be the following solution
\[
\wh{\Theta}_{\tl_1^i} = \arg \inf_{\Theta} \sum_{(t, k) \in \Nc_{\tl_1^i}}||y_k^{(t)} - \Theta u_k^{(t)}||_F^2
\]
where $\{u^{(t)}_k\}_{t, k=1}^{\infty}$ are i.i.d isotropic Gaussian (or subGaussian) random variables. Then whenever $\sqrt{N_{\tl_1^i}} \geq c \prn*{\sqrt{m}  + \sqrt{N\log{(6s/\delta)}}}$ we have with probability at least $1 -\delta$ that
\begin{equation}
||CA_{\tl_1^i}B - \wh{\Theta}_{\tl_1^i}||_F \leq 10 \min{(\sqrt{p}, \sqrt{m})} \cdot{} \beta \sqrt{\frac{N}{N_{m_l}} }\cdot{} \prn*{\log{\frac{3s_N}{\delta}} + p\log{\prn*{5 \beta N}} + m}   .
\end{equation}
Here $N$ is the rollout length and $c$ is an absolute constant.
\end{proposition}
\begin{proof}
Recall $y_t$ as a function of $\{u_l\}_{l=1}^t$, 
\begin{align}
y_t &= \underbrace{\sum_{j=2}^{t-1} CA_{\theta_{t-1}}A_{\theta_{t-2}}\hdots A_{\theta_{j}} B u_{j-1} + CB u_{t-1}}_{\text{Input driven terms}} + \underbrace{\sum_{j=2}^{t-1} CA_{\theta_{t-1}}A_{\theta_{t-2}}\hdots A_{\theta_j} \eta_{j-1}  + C\eta_{t-1} + w_t}_{\text{Noise terms}}, \label{io_rel_2}
\end{align}
then recall that $\Nc_{m_l} \coloneqq \{(t, k)| (\theta^{(t)}_{k+l-1}, \theta^{(t)}_{k+l-2}, \hdots, \theta^{(t)}_{k}) = m_l \in [s]^l\}$ for a switch sequence $m_l \in [s]^l$. According to our experiment set up we collect all data points corresponding to the switch sequence $m_l$. Then we pick $(t, k) \in \Nc_{m_l}$ and 
\begin{equation}
\wh{\Theta}_{m_l} =  \prn*{\sum_{(t, k) \in \Nc_{m_l}} y_{k+l}^{(t)}(u^{(t)}_{k-1})^{\top}} \prn*{\sum_{(t, k) \in \Nc_{m_l}} u^{(t)}_{k-1}(u^{(t)}_{k-1})^{\top}}^{+}\label{eq:theta_hat}
\end{equation}
Since the switch sequences are independent from the input, we have from Proposition~\ref{isometry} we get that whenever $\sqrt{N_{m_l}} \geq c \prn*{\sqrt{m}  + \sqrt{\log{(2/\delta)}}}$ it holds with probability at least $1-\delta$ that 
\begin{equation}
(3/4) \cdot{} N_{m_l} I \preceq \sum_{(t, k) \in \Nc_{m_l}} u^{(t)}_{k-1}(u^{(t)}_{k-1})^{\top} \preceq (5/4) \cdot{} N_{m_l} I \label{eq:isometry}
\end{equation}
Then in Eq.~\eqref{eq:theta_hat} we have for $y_{k+l}^{(t)}(u^{(t)}_{k-1})^{\top}$ that 
\begin{align*}
\sum_{(t, k)} y_{k+l}^{(t)}(u^{(t)}_{k-1})^{\top} &= \sum_{(t, k)} CA_{\tl_1^{m_l}}B  u^{(t)}_{k-1}(u^{(t)}_{k-1})^{\top} + \underbrace{\sum_{(t,k)}\prn*{\sum_{p \neq k-1}\star \cdot{} u^{(t)}_{p}} (u^{(t)}_{k-1})^{\top}}_{=\text{Cross terms}} \\
&+ \underbrace{\sum_{(t,k)}\prn*{\sum_{p \neq k-1}\star \cdot{} \eta^{(t)}_p} (u^{(t)}_{k-1})^{\top}}_{=\text{Noise terms}}.
\end{align*}
We would like to show that the cross terms do not grow more than $\Oc(\sqrt{N_{m_l}})$. Now, $\prn*{\sum_{p \neq k-1}\star \cdot{} u^{(t)}_p} = \sum_{l=1}^N \theta_l u_l$ and the noise term $\prn*{\sum_{p \neq k-1}\star \cdot{} \eta^{(t)}_p} = \sum_{l=1}^N \phi_l \eta_l$. This is because the rollout length is at most $N$ and each of $\eta_l, u_l$ are independent of each other. From Theorem~\ref{subg-conc} we can conclude with probability at least $1-\delta$ that 
\[
\nrm*{\prn*{\sum_{p \neq k-1}\star \cdot{} u^{(t)}_p}}_2 \leq 2 \beta \sqrt{N} \prn*{1 + \sqrt{\log{(1/\delta)}}}, \quad{} \nrm*{\prn*{\sum_{p \neq k-1}\star \cdot{} \eta^{(t)}_p}}_2 \leq 2 \beta \sqrt{N} \prn*{1 + \sqrt{\log{(1/\delta)}}},
\]
where $\sum_{l=1}^N ||\phi_l||_F^2 \leq \beta^2 N$ and $\sum_{l=1}^N ||\theta_l||_F^2 \leq \beta^2 N$. Note that $\phi_l = CA_{\tl_1^j}B$ and $\theta_l = CA_{\tl_1^j}$ for some $\tl_1^j$.

Since $u_p$ is independent of $u_q$ when $p \neq q$ and $\eta_p$ is independent of $u_q$, we can now use Proposition~\ref{selfnorm_main} and these upper bounds above to upper bound the norms of cross and noise terms. In that notation the term inside the parentheses behaves as $X_s$ and $\eta_{s+1}^{\top} = (u^{(t)}_{k-1})^{\top}$. Then with probability at least $1-\delta$, $\sum_{s=1}^{N_{m_l}} X_s X_s^{\top} \preceq 8 \beta^2 N  \prn*{1 + \log{(1/\delta)}} \cdot{}N_{m_l} I$. Note that $V_t  = \sum_{s=1}^t X_s X_s^{\top}$ is a $p \times p$ matrix, then by choosing $V = I$ we have that 
\[
\nrm*{(V + V_t)^{-1/2}S_t}_2  \leq  \sqrt{2 \prn*{\log{\frac{1}{\delta}} + p\log{\prn*{5 \beta N}} + m}}.
\]
Here $S_t = \sum_{k=1}^t\prn*{\sum_{p \neq k-1}\star \cdot{} u^{(t)}_{p}} (u^{(t)}_{k-1})^{\top}$ or $\sum_{k=1}^t\prn*{\sum_{p \neq k-1}\star \cdot{} \eta^{(t)}_p} (u^{(t)}_{k-1})^{\top}$. Since $\nrm*{(V + V_t)^{-1/2}S_t}_2 \geq \sigma_{\min}\prn*{(V + V_t)^{-1/2}} \nrm*{S_t}_2$, we have 
\[
\nrm*{S_{N_{m_l}}}_2 \leq  4  \beta \sqrt{N \cdot{} N_{m_l}}\cdot{} \sqrt{\prn*{\log{\frac{1}{\delta}} + p\log{\prn*{5 \beta N}} + m}\cdot{}\prn*{1 + \log{\frac{1}{\delta}}}}.
\]
Finally, since $\sum_{(t, k) \in \Nc_{m_l}} u^{(t)}_{k-1}(u^{(t)}_{k-1})^{\top}$ satisfies Eq.~\eqref{eq:isometry} we have that 
\[
\nrm*{S_{N_{m_l}} \prn*{\sum_{(t, k) \in \Nc_{m_l}} u^{(t)}_{k-1}(u^{(t)}_{k-1})^{\top}}^{+}}_2 \leq 5 \beta \sqrt{\frac{N}{N_{m_l}} }\cdot{} \sqrt{\prn*{\log{\frac{1}{\delta}} + p\log{\prn*{5 \beta N}} + m}\cdot{}\prn*{1 + \log{\frac{1}{\delta}}}}
\]
with probability at least $1-3\delta$ whenever $\sqrt{N_{m_l}} \geq c \tau^2\prn*{\sqrt{m}  + \sqrt{\log{(2/\delta)}}}$. We need to ensure this for every $s_N = \frac{s^{N+1}-1}{s-1}$ sequences, and therefore using union rule for each sequence and substituting $\delta \rightarrow s_N^{-1} \delta/3$ we get for any $m_l \in \{\tl_1^l| \tl_1^l \in [s]^l \quad{} l \leq N\}$ sequence we have with probability at least $1-\delta$,
\[
\nrm*{S_{N_{m_l}} \prn*{\sum_{(t, k) \in \Nc_{m_l}} u^{(t)}_{k-1}(u^{(t)}_{k-1})^{\top}}^{+}}_2 \leq 5 \beta \sqrt{\frac{N}{N_{m_l}} }\cdot{} \sqrt{\prn*{\log{\frac{3s_N}{\delta}} + p\log{\prn*{5 \beta N}} + m}\cdot{}\prn*{1 + \log{\frac{3s_N}{\delta}}}},
\]
whenever 
$$\sqrt{N_{m_l}} \geq c \prn*{\sqrt{m}  + \sqrt{\log{(6s_N/\delta)}}}.$$ 
The final bound follows from the inequality $2\sqrt{ab} \leq a + b$.
\end{proof}
\subsection{Upper bounds for Estimation error and Truncation error}
\label{sec:ub_est_trunc_err}
\begin{proposition}
\label{error_diff}
Fix $0 \leq N \leq N_{up}$, then for $k \leq N$ we have with probability at least $1-\delta$
\[
\sum_{\tl_{1}^k \in [s]^k}\Big(\sqrt{p_{\tl_1^k}} - \sqrt{\hp_{\tl_1^k}} \Big)^2 ||CA_{\tl_1^k}B||_F^2 \leq \frac{5\beta^2}{N_S} \ts \cdot{} s^{k}.
\]
Here $\ts = \log{\Big(\frac{s^{N+1}-1}{(s-1)\delta}\Big)}$.
\end{proposition}
\begin{proof}
Let $\ts = \log{\Big(\frac{s^{N+1}-1}{(s-1)\delta}\Big)}$. Now we break the sum in two parts 
\begin{align}
    \sum_{\tl_{1}^k \in [s]^k}\Big(\sqrt{p_{\tl_1^k}} - \sqrt{\hp_{\tl_1^k}} \Big)^2 ||CA_{\tl_1^k}B||_F^2 &\leq \underbrace{\sum_{p_{\tl_1^k}N_S \leq \ts}\Big(\sqrt{p_{\tl_1^k}} - \sqrt{\hp_{\tl_1^k}} \Big)^2 ||CA_{\tl_1^k}B||_F^2}_{(i)} \nonumber \\
    &+ \underbrace{\sum_{p_{\tl_1^k}N_S > \ts}\Big(\sqrt{p_{\tl_1^k}} - \sqrt{\hp_{\tl_1^k}} \Big)^2 ||CA_{\tl_1^k}B||_F^2}_{(ii)}
\end{align}
For $(i)$, combine $\prn*{\sqrt{p_{\tl_1^k}} - \sqrt{\hp_{\tl_1^k}}}^2 \leq |{p_{\tl_1^k}} -{\hp_{\tl_1^k}}|$ and use Proposition~\ref{hoeffding} for the case when $p_{\tl_1^k}N_S \leq \ts$ which gives us 
\begin{equation}
    \label{eq:hoeff1}
    (i) \leq \sum_{\tl_{1}^k \in [s]^k} \frac{\log{\prn*{s_N/\delta}} \cdot ||CA_{\tl_1^k}B||_F^2}{N_S} = \sum_{\tl_{1}^k \in [s]^k} \frac{\ts ||CA_{\tl_1^k}B||_F^2}{N_S}.
\end{equation}
For $(ii)$, it follows from the second part in Proposition~\ref{hoeffding} that $\prn*{\sqrt{p_{\tl_1^k}} - \sqrt{\hp_{\tl_1^k}}}^2 \leq p_{\tl_1^k}\prn*{\sqrt{1 + \sqrt{\frac{\ts}{p_{\tl_1^k N_S}}}}-1}^2 \leq \frac{4\ts}{N_S}$. Then we get that $(ii) \leq \sum_{\tl_1^k \in [s]^k} \frac{4\ts ||CA_{\tl_1^k}B||_F^2}{N_S}$. By assumption we have $\sum_{\tl_1^k \in [s]^k} ||CA_{\tl_1^k}B||_F^2 \leq \beta^2 s^k$ and
\[
(i) + (ii) \leq 5\beta^2 \ts /N_S \cdot{} s^k.
\]
\end{proof}
\begin{proposition}
\label{prop:estimation_error}
Fix the rollout length $N$, $N_S$ and $0 < \delta < 1$. Then with probability at least $1 - \delta$ we have
\[
E^2_N = ||\Hc^{(N)} - \wh{\Hc}^{(N)}||_F^2 \leq 2 N^2\frac{\beta^2\mu^2(N)}{N_{S}}\Big(\frac{s^{N+1}-1}{s-1}\Big) = \beta^2 \alpha^2(N)
\]
where $\mu(N) = \sqrt{N}\prn*{N \log{(3s/\delta)} + p\log{\prn*{5 \beta N}} + m}$.
\end{proposition}
\begin{proof}
Let $\ts = \log{\frac{(s^{N+1}-1)}{(s-1)\delta}}$. By definition we have 
\begin{align}
    [\HN-\hHN ]_{L(\tl_1^i), L(\tl_1^j)} &= \sqrt{p_{\tl_1^i:\tl_1^j}}CA_{\tl_1^i:\tl_1^j}B - \sqrt{\hp_{\tl_1^i:\tl_1^j}}\wh{\Theta}_{\tl_1^i:\tl_1^j} \nonumber \\
    &= \underbrace{\Big(\sqrt{p_{\tl_1^i:\tl_1^j}} - \sqrt{\hp_{\tl_1^i:\tl_1^j}} \Big)CA_{\tl_1^i:\tl_1^j}B}_{=\Ec_{1, \tl_{1}^i:\tl_1^{j}}} + \underbrace{\sqrt{\hp_{\tl_1^i:\tl_1^j}}(CA_{\tl_1^i:\tl_1^j}B - \wh{\Theta}_{\tl_1^i:\tl_1^j})}_{=\Ec_{0, \tl_{1}^i:\tl_1^{j}}} \nonumber 
\end{align}
First we analyze $\Ec_{1, \tl_{1}^i:\tl_1^{j}}$. It is clear that 
\begin{align*}
 \sum_{0 \leq i+j \leq N}||\Ec_{1, \tl_{1}^i:\tl_1^{j}}||_F^2 &= \sum_{0 \leq i+j \leq N}\Big(\sqrt{p_{\tl_1^i:\tl_1^j}} - \sqrt{\hp_{\tl_1^i:\tl_1^j}} \Big)^2 ||CA_{\tl_1^i:\tl_1^j}B||_F^2   \\
  &\leq \sum_{k=0}^{N} (k+1)\sum_{\tl_1^k \in [s]^k}\Big(\sqrt{p_{\tl_1^k}} - \sqrt{\hp_{\tl_1^k}} \Big)^2 ||CA_{\tl_1^k}B||_F^2 
\end{align*}
The $k+1$ is the number of times a $k$--length sequence appears in the Hankel--like matrix. Using Proposition~\ref{error_diff} with probability at least $1-\delta$, $\sum_{0 \leq i+j \leq N}||\Ec_{1, \tl_{1}^i:\tl_1^{j}}||_F^2 \leq \frac{10 N^2 \beta^2\ts}{N_S} \cdot{}s_N$. For $\Ec_{0, \tl_1^i:\tl_1^{j}}$ we have
\begin{align*}
    \sum ||\Ec_{0, \tl_1^i:\tl_1^{j}}||_F^2 \leq  \sum_{0 \leq i+j \leq N} {{\hp_{\tl_1^i:\tl_1^j}}||CA_{\tl_1^i:\tl_1^j}B - \wh{\Theta}_{\tl_1^i:\tl_1^j}||_F^2}
\end{align*}
Recall that whenever $N_{\tl_1^{i}:\tl_{1}^j} <  \ts$, \textit{i.e.}, scarce data, we set $\wh{\Theta}_{\tl_1^i:\tl_1^j} = 0$. Then we get 
\begin{align}
    \sum {{\hp_{\tl_1^i:\tl_1^j}}||CA_{\tl_1^i:\tl_1^j}B - \wh{\Theta}_{\tl_1^i:\tl_1^j}||_F^2} \leq \sum_{N_{\tl_1^{i}:\tl_{1}^j} < \ts} \hp_{\tl_1^i:\tl_1^j}||CA_{\tl_1^i:\tl_1^j}B||_F^2 + \sum_{N_{\tl_1^{i}:\tl_{1}^j} \geq \ts} \hp_{\tl_1^i:\tl_1^j} ||CA_{\tl_1^i:\tl_1^j}B - \wh{\Theta}_{\tl_1^i:\tl_1^j}||_F^2 \label{err0}
\end{align}
We now use Proposition~\ref{regression_result} (applied with union bound to all sequences) with probability at least $1-\delta$
\begin{align}
\sum_{N_{\tl_1^{i}:\tl_{1}^j} \geq \ts} \frac{N_{\tl_1^i:\tl_1^j}||CA_{\tl_1^i:\tl_1^j}B - \wh{\Theta}_{\tl_1^i:\tl_1^j}||_F^2}{N_S(N-i-j+1)} &\leq 100\min{(p, m)} \cdot{} \sum_{N_{\tl_1^{i}:\tl_{1}^j} \geq \ts} \frac{\beta^2 \overbrace{N\prn*{N \log{(3s/\delta)} + p\log{\prn*{5 \beta N}} + m}^2}^{=\mu^2(N)}}{N_{S}(N-i-j+1)}\\
&\leq  \sum_{N_{\tl_1^{i}:\tl_{1}^j} \geq \ts}  \frac{\beta^2\mu^2(N)}{N_{S}(N-i-j+1)}   \leq   \sum_{k=0}^{N}(k+1)\sum_{\tl_1^k \in [s]^k}  \frac{\beta^2\mu^2(N)}{N_{S}(N-k+1)}  \nonumber\\
    &\leq N^2 \sup_{k \leq N}\sum_{\tl_1^k \in [s]^k}  \frac{\beta^2s^{k}\mu^2(N)}{N_{S}} \leq  N^2\frac{\beta^2\mu^2(N)}{N_{S}} \cdot{}s_N  \label{err1}
\end{align}
From first part in Proposition~\ref{hoeffding} we get with probability at least $1-\delta$
\begin{align}
\sum_{N_{\tl_1^{i}:\tl_{1}^j} < \ts}{||CA_{\tl_1^i:\tl_1^j}B||_F^2 \hp_{\tl_1^{i}:\tl_{1}^j}} \leq \frac{ N^2\ts}{N_S} \sup_{k}\sum_{\tl_1^k \in [s]^k} ||CA_{\tl_1^k}B||_F^2 \leq \frac{\beta^2 N^2\ts}{N_S} \cdot{}s_N
\end{align}
Then combining these observations and the fact that $(m + \ts) \leq \mu^2(N)$ we get with probability at least $1-\delta$
\[
E^2_N \leq 2 N^2\frac{\beta^2\mu^2(N)}{N_{S}}\cdot{}s_N.
\]
\end{proof}
\begin{proposition}
\label{truncation_v_estimation}
Recall that $\bar{\Hc}^{(N)}$ is the zero padded version of $\Hc^{(N)}$ to make it compatible with $\Hci$. Let $T_{N}^2 = {||\bar{{\Hc}}^{(N)} - \Hci||_F^2}$. Then 
\[
T_N^2 \leq N \cdot{}\frac{ c(n) \rho_{\max}^{N+1}}{1-\rho_{\max}}
\]
$\rho_{\max} = \rho \prn*{\sum_{i=1}^s p_i A_i \otimes A_i}$ and $c(n)$ depends only $m, n, p$.
\end{proposition}
\begin{proof}

Since the SLS is mean-square stable and by our assumptions we have
\begin{align*}
 ||\bar{\Hc}^{(N)} - \Hci||_F^2 &=\sum_{i+j \geq N+1} p_{\tl_1^i:\tl_1^{j}}||CA_{\tl_1^i:\tl_1^{j}}B||_F^2 =  \sum_{k \geq N+1} p_{\tl_1^k}\sum_{\tl_1^k \in [s]^k } (k+1)||CA_{\tl_1^k}B||_F^2\\
 &= \sum_{k \geq N+1} \sum_{\tl_1^k \in [s]^k } (k+1)p_{\tl_1^k}||CA_{\tl_1^k}B||_F^2 \leq \frac{N c(n) \rho_{\max}^{N+1}}{1-\rho_{\max}},
\end{align*}
where $\rho_{\max} = \rho \prn*{\sum_{i=1}^s p_i A_i \otimes A_i}$ and $c(n)$ depends only $m, n, p$. 

\end{proof}

Unfortunately, the upper bound on $T_N$ is apriori unknown due to its dependence on $\rho_{\max}$ and $n$. Consequently, the balancing suggested in Eq.~\eqref{eq:balancing_1} and \eqref{eq:balancing} cannot be achieved using the upper bound mentioned in Proposition~\ref{truncation_v_estimation}.


Define $\est_N^2 \coloneqq 2 N^2\frac{\beta^2\mu^2(N)}{N_{S}}\cdot{}s_N$ and $\trunc_N^2 \coloneqq N \cdot{}\frac{ c(n) \rho_{\max}^{N+1}}{1-\rho_{\max}}$.
\begin{proposition}
\label{prop:N_exist}
Let
\[
N_S \geq \prn*{\frac{N_{up}\cdot{} c(n) \rho_{\max}}{(1-\rho_{\max}) \cdot{}\log{N_S}}}^{\frac{\log{\prn*{1/p_{\max}}}}{\log{\prn*{1/\rho_{\max}}}}}
\]
where $c(n)$ is a constant that depends on $n, m, p$ only. Then there exists $N\leq N_{up}$ such that 
\[
\trunc_N \leq \est_N
\]
with probability at least $1-\delta$.
\end{proposition}
\begin{proof}
Let $\rho = \rho_{\max}$. For $N = N_{up}$ we have the estimation error as 
\begin{align*}
    \est_N^2 &= 2N^2\beta^2\frac{\mu^2(N)}{N_{S}}\cdot{}s_N = 2N_{up}^2 \frac{\beta^2\mu^2(N_{up})}{N_{S}}\cdot{}s_{N_{up}}.
\end{align*}
Then $s_{N_{up}} \geq s^{\frac{\log{N_S}}{\log{\frac{1}{p_{\max}}}}}$ from Proposition~\ref{prop:n_up} with probability at least $1-\delta$ and 
\[
s^{\frac{\log{N_S}}{\log{\frac{1}{p_{\max}}}}} = \prn*{sp_{\max}}^{\frac{\log{N_S}}{\log{\frac{1}{p_{\max}}}}} \prn*{\frac{1}{p_{\max}}}^{\frac{\log{N_S}}{\log{\frac{1}{p_{\max}}}}} \geq N_S.
\]
It follows then that 
\[
\est_N^2 \geq \log{N_S}.
\]
For the truncation error, $\trunc_N$ observe that since $N_{up} \geq \frac{\log{N_S}}{\log{\frac{1}{p_{\max}}}}$, we have that 
\[
\rho^{N_{up}} \leq N_S^{-\frac{\log{\prn*{1/\rho}}}{\log{\prn*{1/p_{\max}}}}}.
\]
Then whenever 
\[
N_S \geq \prn*{\frac{N_{up}\cdot{} c(n) \rho_{\max}}{(1-\rho_{\max}) \cdot{}\log{N_S}}}^{\frac{\log{\prn*{1/p_{\max}}}}{\log{\prn*{1/\rho_{\max}}}}}
\]
it is ensured that $\trunc_N^2 \leq \log{N_S}$.
\end{proof}

\begin{proposition}
\label{prop:n_star}
Let
\[
N_S \geq \prn*{\frac{N_{up}\cdot{} c(n) \rho_{\max}}{(1-\rho_{\max}) \cdot{}\log{N_S}}}^{\frac{\log{\prn*{1/p_{\max}}}}{\log{\prn*{1/\rho_{\max}}}}}
\]
where $c(n)$ is a constant that depends on $n, m, p$ only. Define $\Ns \coloneqq \inf{\curly*{N| \quad{} \est_N \geq \trunc_N }}$. Then we have with probability at least $1-\delta$ that 
\[
\est_{\Ns}^2 =\wt{\Oc}\prn*{ N_S^{-\Delta_s}},
\]
where $\rho_{\max} = \rho\prn*{\sum_{i=1}^s p_i A_i \otimes A_i}$ and $\Delta_s = \frac{\log{\prn*{1/\rho_{\max}}}}{\log{\prn*{s/\rho_{\max}}}} > 0$.
\end{proposition}
\begin{proof}
The condition on $N_S$ ensures that $\Ns$ exists from Proposition~\ref{prop:N_exist}. We first find the $N$ for which $\est_N = \trunc_N$,
\begin{align*}
    N^2 \beta^2 \frac{\mu^{2}(N)}{N_S} \cdot{} s_N &= N \cdot{} \frac{\rho_{\max}^{N+1}}{1 -\rho_{\max}}, \\
    \implies \frac{N \mu^2(N)}{N_S} \cdot{}\wt{c}(n) &= \prn*{\frac{\rho_{\max}}{s}}^{N+1},
\end{align*}
where $\wt{c}(n) = \frac{\beta^2 (m+s_0)(1-\rho_{\max})}{c(n) s}$. This gives us that 
\[
\Ns+1 = \frac{\log{N_S} - \log{(\mu^2(\Ns)\Ns \wt{c}(n))}}{\log{\frac{s}{\rho_{\max}}}}.
\]
Then computing $\trunc_{\Ns}$ we get
\begin{align*}
    \trunc_{\Ns}^2 &= \Ns \cdot{}\frac{\rho_{\max}^{\Ns}+1}{1 - \rho_{\max}} = \Ns\cdot{} \frac{\rho_{\max}^{\frac{\log{N_S}}{\log{\frac{s}{\rho_{\max}}}}} \cdot{} \rho_{\max}^{-\frac{\mu^2(\Ns)\Ns \wt{c}(n)}{\log{\frac{s}{\rho_{\max}}}}}}{1 - \rho_{\max}},\\
    &= (\mu^2(\Ns))^{-\Delta_s} \cdot{} \wt{c}(n)^{-\Delta_s} \cdot{} (\Ns)^{1 - \Delta_s} \cdot{} N_S^{-\Delta_s},
\end{align*}
where $\Delta_s = \frac{\log{\prn*{1/\rho_{\max}}}}{\log{\prn*{s/\rho_{\max}}}}$. Note that since $\Ns$ depends only logarithmically on $N_S$ the term $N_S^{-\Delta_s}$ dominates.
\end{proof}

\begin{remark}
\label{rem:delta_s}
When $s = 1$, note that $\Delta_s = 1$ and we get the estimation error rates for linear time invariant (LTI) systems.
\end{remark} 

\subsection{Model Selection Results}
\label{sec:model_selection}
Recall the following definitions:
\[
\est_N = 2 N^2\frac{\beta^2\mu^2(N)}{N_{S}}\cdot{}s_N, \quad{} \trunc_N = N \cdot{}\frac{ c(n) \rho_{\max}^{N+1}}{1-\rho_{\max}}.
\]
From Proposition~\ref{prop:N_exist}, whenever 
\[
N_S \geq \prn*{\frac{N_{up}\cdot{} c(n) \rho_{\max}}{(1-\rho_{\max}) \cdot{}\log{N_S}}}^{\frac{\log{\prn*{1/p_{\max}}}}{\log{\prn*{1/\rho_{\max}}}}}.
\]
there exists $\Ns < \infty$ where 
\begin{equation}
    \label{def:ns}
\Ns \coloneqq \inf{\curly*{N | \est_N \geq \trunc_N}}.
\end{equation}
Finally, 
\begin{equation}
\hN \coloneqq  \inf\Bigg\{l \Bigg| ||\hHc^{(d)} - \hHc^{(l)}||_F \leq \beta(\alpha(d) + 2 \alpha(l)) \quad{} \forall N_{up} \geq d \geq l\Bigg\} \label{eq:hd_eq_2}    
\end{equation}
where $\alpha(\cdot)$ is defined in Table~\ref{notation}.
\begin{proposition}
\label{prop:ns_hn}
Assume that 
\[
N_S \geq \prn*{\frac{N_{up}\cdot{} c(n) \rho_{\max}}{(1-\rho_{\max}) \cdot{}\log{N_S}}}^{\frac{\log{\prn*{1/p_{\max}}}}{\log{\prn*{1/\rho_{\max}}}}}.
\]
Then with probability at least $1-\delta$, we have 
\[
\hN \leq \Ns.
\]
\end{proposition}
\begin{proof}
For any $d \geq \Ns$ we have that
\begin{align*}
||\wh{\Hc}^{(d)} - \wh{\Hc}^{(\Ns)}||_F &\leq ||\wh{\Hc}^{(d)} - {\Hc}^{(d)}||_F + ||{\Hc}^{(d)} - {\Hc}^{(\Ns)}||_F + ||\wh{\Hc}^{(\Ns)} - {\Hc}^{(\Ns)}||_F, \\
&\leq ||\wh{\Hc}^{(d)} - {\Hc}_{d}||_F + ||{\Hc}^{(\Ns)} - {\Hc}^{(\infty)}||_F + ||\wh{\Hc}^{(\Ns)} - {\Hc}^{(\Ns)}||_F
\end{align*}
then from Proposition~\ref{prop:estimation_error} and definition of $\Ns$ we get that 
\[
||\wh{\Hc}^{(d)} - \wh{\Hc}^{(\Ns)}||_F \leq \beta(2 \alpha(\Ns) + \alpha(d))
\]
with probability at least $1-\delta$. This shows that $\hN \leq \Ns$ since it is the least $N$ that satisfies this condition.
\end{proof}

We know from Proposition~\ref{prop:n_star} that 
\[
||\Hc^{(\Ns)} - \Hc^{(\infty)}||_F \leq \est_{\Ns} + \trunc_{\Ns} = \wt{\Oc}(N_S^{-\Delta_s/2}).
\]
We will show that $\hN$ also satisfies a similar property, \textit{i.e.}, $||\Hc^{(\hN)} - \Hc^{(\infty)}||_F = \wt{\Oc}(N_S^{-\Delta_s/2})$. A key property to prove this result will be that $\hN \leq \Ns$.

\begin{proposition}
\label{prop:ns_hn}
Assume 
\[
N_S \geq \prn*{\frac{N_{up}\cdot{} c(n) \rho_{\max}}{(1-\rho_{\max}) \cdot{}\log{N_S}}}^{\frac{\log{\prn*{1/p_{\max}}}}{\log{\prn*{1/\rho_{\max}}}}}.
\]
Then with probability at least $1-\delta$, we have 
\[
||\wh{\Hc}^{(\hN)} - \Hc^{(\infty)}||^2_F = \wt{\Oc}(N_S^{-\Delta_s})
\]
where $\Delta_s = \frac{\log{\prn*{1/\rho_{\max}}}}{\log{\prn*{s/\rho_{\max}}}}$.
\end{proposition}
\begin{proof}
For any $\hN$ we have that
\begin{align*}
||\Hci - \wh{\Hc}^{(\Ns)}||_F &\leq ||\wh{\Hc}^{(\hN)} - {\Hc}^{(\Ns)}||_F + ||{\Hc}^{(\Ns)} - \Hci||_F \\
&\leq ||\wh{\Hc}^{(\hN)} - \wh{\Hc}^{(\Ns)}||_F + ||{\Hc}^{(\Ns)} - \wh{\Hc}^{(\Ns)}||_F + ||{\Hc}^{(\Ns)} - \Hci||_F\\
&\leq \beta (2 \alpha(\hN) + \alpha(\Ns)) + 2\beta \alpha(\Ns).
\end{align*}
Since $\hN \leq \Ns$, $\alpha(\hN) \leq \alpha(\Ns)$ and we have $||\Hci - \wh{\Hc}^{(\Ns)}||_F \leq 5 \beta \alpha(\Ns) = \wt{\Oc}(N_S^{-\Delta_s/2})$ from Proposition~\ref{prop:n_star}.
\end{proof}

\subsection{Proof of Theorem~\ref{balanced_truncation}}
\label{sec:thm_balanced_proof}
The proof follows from Proposition~\ref{prop:ns_hn} and then using Proposition~\ref{prop:c_select}. We let $\epsilon = \wt{\Oc}(N^{-\Delta_s/2})$. Define
\[
\Gamma(\Sigma, \epsilon) = \sqrt{ {\sigma}_{1}/\zeta_{n_1}^2 + {\sigma}_{n_1+1}/\zeta_{n_2}^2 + \hdots + {\sigma}_{\sum_{i=1}^{l-1}n_i+1}/\zeta_{n_l}^2}
\]
and by Proposition~\ref{prop:c_select} we have
	\begin{align*}
	\max{ \prn*{||\wh{C}^{(r)} - C^{(r)}Q||_2, ||\wh{B}^{(r)} - Q^{\top}B^{(r)}||_2}} &\leq 2 \epsilon \Gamma(\Sigma, \epsilon) + 2\sup_{1 \leq i \leq l}\sqrt{\sigma^i_{\max}} - \sqrt{\sigma^i_{\min}} \\
	&+ \frac{\epsilon}{\sqrt{\sigma_{r}}} \wedge \sqrt{\epsilon} = \zeta,\\
	||Q^{\top}A^{(r)}Q - \wh{A}^{(r)}||_2 &\leq  4 \gamma \cdot \zeta / \sqrt{\sigma_{r}}.
	\end{align*}
	Here $\sup_{1\leq i \leq l}\sqrt{\sigma^i_{\max}} - \sqrt{\sigma^i_{\min}} \leq \frac{\chi }{\sqrt{\sigma^i_{\max}}} \epsilon r \wedge \sqrt{\chi r \epsilon}$ and 
	\[
	\zeta_{n_i} = \min{({\sigma}^{n_{i-1}}_{\min}-{\sigma}^{n_i}_{\max}, {\sigma}^{n_i}_{\min}-{\sigma}^{n_{i+1}}_{\max})}
	\]
	for $1 < i < l$, $\zeta_{n_1} ={\sigma}^{n_1}_{\min}-{\sigma}^{n_2}_{\max}$ and $\zeta_{n_l} = \min{({\sigma}^{n_{l-1}}_{\min}-{\sigma}^{n_l}_{\max}, {\sigma}^{n_l}_{\min})}$. If $\sigma_r = \Omega(N_S^{-\Delta_s/2})$, \textit{i.e.}, $\frac{\chi }{\sqrt{\sigma^i_{\max}}} \epsilon r \geq \sqrt{\chi r \epsilon}$, then we have that 
	\begin{align*}
	\max{ \prn*{||\wh{C}^{(r)} - C^{(r)}Q||_2, ||\wh{B}^{(r)} - Q^{\top}B^{(r)}||_2}} &= \wt{\Oc}\prn*{\frac{\Gamma(\Sigma, \epsilon) \cdot{} N_S^{-\Delta_s/2}}{\sqrt{\sigma_r}} \vee r N_S^{-\Delta_s/2}} ,\\
	||Q^{\top}A^{(r)}Q - \wh{A}^{(r)}||_2 &= \wt{\Oc}\prn*{\frac{\gamma \Gamma(\Sigma, \epsilon) \cdot{} N_S^{-\Delta_s/2}}{\sqrt{\sigma_r}} \vee r N_S^{-\Delta_s/2}}.
	\end{align*}	

Since $\Sigma$ is finite dimensional the term $\Gamma(\Sigma, \epsilon) < \infty$ for all $\epsilon \geq 0$. It is sufficient to check for $\epsilon = 0$ and $\epsilon = \infty$. For $\epsilon = \infty$, $\Gamma(\Sigma, \infty)$ simply equals $\sqrt{\frac{\sigma_1}{\sigma_r^2}}$. For $\epsilon=0$, then $\Gamma(\Sigma, 0) = \sqrt{\frac{\sigma_1}{\Delta_+^2}}$ where 
\[
\Delta_+ = \min_{\sigma_{i} \neq \sigma_{i+1}} \sigma_i - \sigma_{i+1}.
\]
Here $\sigma_i = \sigma_i\prn*{\Hci}$. For the case $\sigma_r = o\prn*{N_S^{-\Delta_s/2}}$, we have 
	\begin{align*}
	\max{ \prn*{||\wh{C}^{(r)} - C^{(r)}Q||_2, ||\wh{B}^{(r)} - Q^{\top}B^{(r)}||_2}} &= \wt{\Oc}\prn*{\frac{\Gamma(\Sigma, \epsilon) \cdot{} N_S^{-\Delta_s/2}}{\sqrt{\sigma_r}} \vee \sqrt{r N_S^{-\Delta_s/2}}} ,\\
	||Q^{\top}A^{(r)}Q - \wh{A}^{(r)}||_2 &=  \wt{\Oc}\prn*{\frac{\gamma \Gamma(\Sigma, \epsilon) \cdot{} N_S^{-\Delta_s/2}}{\sqrt{\sigma_r}} \vee \sqrt{r N_S^{-\Delta_s/2}}}.
	\end{align*}



\end{document}